\documentclass{elsart}

\usepackage{subfigure}
\usepackage[dvips]{graphicx}
\usepackage{color}

\usepackage{amssymb}

\usepackage{moreverb}

\newcommand{\point}[1]{{\mathbf {#1}}}
\newcommand{\vect}[1]{{\overrightarrow{\mathbf {#1}}}}
\newcommand{\llangle}{\langle\langle}
\newcommand{\rrangle}{\rangle\rangle}

\newtheorem{definition}{Definition}[section]
\newtheorem{theorem}{Theorem}[section]
\newtheorem{corollary}{Corollary}[section]
\newtheorem{lemma}{Lemma}[section]
\newtheorem{proposition}{Proposition}[section]

\newcounter{remctr}[section]
\newenvironment{remark}
{\stepcounter{remctr}%
 {\bf {\em Remark \arabic{section}.\arabic{remctr}}} \ }{$\blacklozenge$\\}
\newenvironment{proof}{{\em Proof.}\ }{$\blacksquare$\\}

\begin{document}
\begin{frontmatter}

\title{Weighted distance transforms\\ generalized to modules\\
 and their computation on point lattices}
\author[CBA_UU]{C\'{e}line Fouard}
\ead{celine@cb.uu.se}
\author[CBA_UU]{Robin Strand\corauthref{cor1}}
\ead{robin@cb.uu.se}
\author[CBA_SLU]{Gunilla Borgefors}
\ead{gunilla@cb.uu.se}
\corauth[cor1]{Corresponding author}
\address[CBA_UU]{Centre for Image Analysis, Uppsala University, \\
L\"{a}gerhyddsv\"{a}gen 3, SE-75237 Uppsala, Sweden} 
\address[CBA_SLU]{Centre for Image Analysis, Swedish University of
Agricultural Sciences, \\ L\"{a}gerhyddsv\"{a}gen 3, SE-75237 Uppsala,
Sweden}

\begin{abstract}
This paper presents the generalization of weighted distances to 
modules and their computation through the chamfer algorithm on general 
point-lattices. A first part is devoted to the definitions and 
properties (distance, metric, norm) of weighted distances on modules, 
with a presentation of weight optimization in the general case, to 
get rotation invariant distances. A general formula for the weighted 
distance on any module is presented. The second part of this paper 
proves that, for any point-lattice, the sequential 2-scan chamfer 
algorithm  produces correct distance maps. 
Finally, the definitions and computation of weighted 
distances are applied to the face-centered cubic (FCC) and 
body-centered cubic (BCC) grids. 
\end{abstract}


\end{frontmatter}

\section{Introduction}
Given a binary image consisting of object and background grid
points, the distance transform assigns to each object grid point its
distance value to the closest background grid point. The distance
transform provides valuable information for shape analysis and is
widely used in several applications \cite{borgefors:avfp:1994}, 
such as, for example, skeleton
extraction \cite{pudney:cviu:1998}, template matching
\cite{borgefors:PAMI:1988}, shape based interpolation
\cite{herman:cga:1992,grevera:tmi:1996} or image registration
\cite{cai:ijrobp:1999}.

Computing the distance from each object grid point to each background
grid point would lead to a far too high computational cost. Numerous
authors have thus investigated alternative ways of computing
distance maps. To do so, they use the spatial consistency of a
distance map, which allows propagation of local information.
In this paper, we focus on one such algorithm, the chamfer algorithm,
which computes the {\em weighted distance transform} (WDT). 
When using the chamfer algorithm, only a
small neighborhood of each grid point is considered. A weight, a
\textit{local distance}, is assigned to each grid point in the
neighborhood. By propagating the local distances in the two-scan
algorithm, the correct distance map is obtained. For example, the
well-known two dimensional city-block ($L^1$) and chessboard
($L^\infty$) distances can be obtained in this way by using unit
weights for the neighbors.

It has been shown that in two dimensions, the hexagonal grid is in
many ways to prefer over the usual square grid,
\cite{bell:ivc:1989}. For example, since the hexagonal grid
constitutes the closest sphere packing \cite{conway:gmw:1988}, only
$87\%$ of the number of samples needed in the square grid can be used
without information loss, \cite{ibanez:dgci:1996}. Extending this to
three dimensions, the face-centered cubic (FCC) grid is the lattice
with the highest packing density, resulting in that less samples are
needed for the FCC grid without loosing information, compared to the
cubic grid, \cite{ibanez:dgci:1996}. The reciprocal grid of the FCC
grid is the body-centered cubic (BCC) grid. This means that if the FCC
grid is used in the spatial/frequency domain, then the BCC grid is
used to represent the image in the frequency/spatial domain. The FCC
grid has the highest packing density, so this grid is preferably used
in frequency domain resulting in that only $71\%$ of the samples are
needed for an image on the BCC grid compared to the cubic grid. If the
FCC grid is used in spatial domain instead, $77\%$ of the samples can
be used without loosing information, \cite{ibanez:dgci:1996}. This
reasoning requires that it is possible to acquire images directly to
the grids with non-cubic voxels. 
In theory, any tomography-based technique for acquiring volume images
(i.e., that gives a sequence of 2D slices) such as Magnetic Resonance
Imaging (MRI), Computed Tomography (CT), Single Photon Emission
Computed Tomography (SPECT), or Positron Emission Tomography (PET) can
be adjusted to work on the FCC or BCC grids.

For SPECT and CT, using
the filtered back-projection method or direct Fourier methods,
the fact that the Fourier transform for the FCC and BCC grids (and for
two-dimensional planes that correspond to the 2D slices) exists is
enough. Images on the FCC and BCC grids acquired using the algebraic
reconstruction technique for SPECT or CT is found in
\cite{sorzano:jsb:2004}. They use spherically-symmetric volume
elements (blobs), \cite{matej:tns:1995}. With the PET-technique, the
origin of gamma-rays (produced by annihilation of an emitted positron
and an electron) are computed. The technique does not depend on the
underlying grid except in the digitization step, which just as well
can be performed on the FCC and BCC grids.

Dealing with digital images means dealing with discrete
grids. Indeed, digital images are discrete representation of the
continuous world. Two kind of approaches can be adopted to deal with
discrete grids. 
The first one applies continuous methods to
the discrete grid. This is the case of the Euclidean distance
transform (EDT) which aims at computing exact Euclidean distance on
discrete images. 
The second approach, called discrete geometry, consists in developing
tools directly devoted to digital images and discrete grids. This is
the case of the WDT which can, for example, be
computed using a two-scan algorithm, a {\em chamfer algorithm} in the
usual square grid \cite{rosenfeld:acm:1966}.

Since computing the WDT is less memory demanding and faster than
computing the EDT, the WDT is preferable in
situations where the data set is too big to fit into the
central memory (such as in \cite{fouard:isbi:2004}) or the
computational time should be minimized (recent 
contributions involve speedup of level-set methods,
\cite{krissian:prl:2005}, and real-time tele-operated force-feedback,
\cite{lee:is:2005}). 
It should be mentioned that there are linear-time (or almost linear)
algorithms for computing the EDT obtained by
sequentially propagating vectors, \cite{danielsson:cgip:1980}, and
dimensionality reduction,
\cite{saito:pr:1994,breu:pami:1995,maurer:pami:2003}. These algorithms
are, however, not as fast as the chamfer algorithm since the EDT
requires a larger number of arithmetic operations. Also, the vector
propagation algorithm is never guaranteed to be error-free,
\cite{forchhammer:SCIA:1989}. 
The WDT has other properties which makes it useful in applications. 
It consists of only integers without loosing the metric property which
is not the case for the EDT (real-valued distances) or the squared EDT
(does not satisfy the triangular inequality). This leads to nice
properties, e.g., when extracting the centers of maximal balls for
computing skeletons which is fast and simple for the most common WDTs,
\cite{arcelli:cvgip:1988,borgefors:avfp:1994}, in contrast to the case
for EDT, \cite{remy:dgci:2003}. Observe that, when using large
neighborhoods in the WDT, then the situation can be complex also for
WDTs, \cite{remy:prl:2001}. In general, the weighted distances are
well-suited for morphological operations, \cite{nacken:jmiv:1994}. 

To allow image processing on images on the non-standard grids,
algorithms must be designed to work directly on these grids. Weighted
distances and the chamfer algorithm has been applied to non-standard
grids such as the two-dimensional hexagonal grid,
\cite{borgefors:cvgip:1984,borgefors:prl:1989,strand:cviu:2005} and
the FCC and BCC grids, \cite{strand:cviu:2005}. In this paper, the
weighted distance and the chamfer algorithm are considered on
\textit{modules}. This gives a very general framework in which
$\mathbb{Z}^n$, the hexagonal, FCC, and BCC grids are all considered
in parallel. The general framework implies that the theory also
applies to the higher-dimensional generalizations of these grids.

This paper is organized as follows: we first summarize weighted
distance definitions and properties that can be found in the
literature and generalize them to sub-modules of
$\mathbb{R}^{n}$. Then we exhibit conditions for the chamfer algorithm
to work on our framework. Finally, we propose examples of chamfer
masks for the body-centered cubic (BCC) grid and the
face-centered cubic (FCC) grid.

\section{Definitions, notations and properties}
This section generalizes definitions and weighted
distance properties found in literature to our general framework of
modules. We denote the set of real numbers, the set of integers, and
the set of natural numbers, with $\mathbb{R}$, $\mathbb{Z}$, and
$\mathbb{N}$, respectively.

\subsection{General framework: Module}

Let $(\mathcal{G}, +)$ be an Abelian group ($\mathcal{G}$ could be, for
example, $\mathbb{R}^{n}$ or $\mathbb{Z}^{n}$ for $n \in
\mathbb{N}\backslash \{0\}$).  
To define weighted distances, we need not only an internal operator
(here $+$), but also an external operation: $\alpha \cdot \point{p}$
with $\alpha \in \mathbb{R}, \mathbb{Z} \textrm{, or } \mathbb{N}$ and
$\point{p} \in \mathcal{G}$.
Vector spaces (as for example $(\mathbb{R}^{n},\mathbb{R}, +,
\times)$) handle external operations, but are too restrictive
($(\mathbb{Z}^{n}, \mathbb{Z}, +, \times)$, for example, is not a
vector space). To be general enough, we use modules.

\begin{definition}[Module]
Let $\mathcal{R}$ be a commutative ring (for example $(\mathbb{Z}, +,
\times)$) with two neutral elements $0$ and $1$. A set $\mathcal{G}$
is called a {\em module on $\mathcal{R}$} (or $\mathcal{R}$-module) if
$\mathcal{G}$ has a commutative group operation $+$, an external law $\cdot$,
and satisfies the following properties:
\begin{eqnarray*}
  \textrm{ (identity) } & \forall \point{p} \in \mathcal{G} & 1 \cdot
  \point{p} = \point{p} \\ 
  \textrm{ (associativity) } & \forall \point{p} \in \mathcal{G} \textrm{, }
  \forall \alpha, \beta \in \mathcal{R}, \ & \alpha \cdot ( \beta
  \cdot \point{p}) = (\alpha \times \beta) \cdot \point{p} \\
  \textrm{ (scalar distributivity) } & \forall \point{p} \in \mathcal{G}
  \textrm{, } \forall \alpha, \beta \in \mathcal{R}, & (\alpha +
  \beta) \cdot \point{p} = \alpha \cdot \point{p} +  \beta \cdot \point{p} \\ 
  \textrm{ (vectorial distributivity) } & \forall \point{p}, \point{q} \in \mathcal{G}
  \textrm{, } \forall \alpha \in \mathcal{R}, & \alpha \cdot (\point{p} + \point{q}) =
  \alpha \cdot \point{p} + \alpha \cdot \point{q}\\
\end{eqnarray*}
\end{definition}

We now consider sub-rings $(\mathcal{R}, +, \times)$ of
$(\mathbb{R}, +, \times)$. Given
an Abelian group $(\mathcal{G}, +)$ we consider the $\mathcal{R}$-module
$(\mathcal{G}, \mathcal{R}, +, \cdot)$. 

\begin{remark}
\label{rem:vspace_module}
The main difference between a module and a vector space is the
non-invertibility (with respect to the external law) of the elements
of its associated ring $\mathcal{R}$ (e.g. $2 \in \mathbb{Z}$ but $1/2
\not\in \mathbb{Z}$). 
A basis of a module $\mathcal{G}$ of dimension $n$ is a
family of $n$ independent vectors $(\vect{v_{i}})_{i=1..n}$ ($\forall
\alpha \in \mathcal{R}^{n} \textrm{, } \sum_{i=1}^{n} \alpha_{i} \cdot
\vect{v_{i}} = 0 \Leftrightarrow  \forall i \in [1..n] \textrm{, }
\alpha_{i} = 0$). 
But a linearly independent family of $n$  vectors may not be a basis
of $\mathcal{G}$. For example, $\left( \overrightarrow{(1,0)},
\overrightarrow{(0,2)} \right)$ is a
family of two independent vectors of $\mathbb{Z}^2$ but is not a
basis of $\mathbb{Z}^2$ ($ \overrightarrow{(1,1)} \in  \mathbb{Z}^2$
can not be reached by a linear combination of $\overrightarrow{(1,0)}$
and $\overrightarrow{(0,2)}$ with coefficients taken in the ring
$\mathbb{Z}$). 
\end{remark}

\begin{definition}[Distance]
A {\em distance} on a group $\mathcal{G}$, having values in
$\mathcal{R}$, called $(d, \mathcal{G}, \mathcal{R})$ is
a function $d\colon \mathcal{G} \times \mathcal{G} \mapsto \mathcal{R}$
which satisfies the following properties:
\begin{eqnarray*}
\textrm{(positive) } & \forall \point{p},\point{q} \in \mathcal{G} &\ 
d(\point{p},\point{q}) \geq 0 \\
\textrm{(definite) } & \forall \point{p},\point{q} \in \mathcal{G} &\ 
d(\point{p},\point{q}) = 0 \Leftrightarrow \point{p}=\point{q} \\
\textrm{(symmetric) }   & \forall \point{p}, \point{q} \in \mathcal{G}
&\ d(\point{p},\point{q}) = d(\point{q},\point{p}) \\ 
\end{eqnarray*}
\end{definition}

\begin{definition}[Metric]
Given a distance $(d, \mathcal{G}, \mathcal{R})$, $d$ is called a {\em
metric} if it also satisfies the following property:
\begin{eqnarray*}
\textrm{(triangular inequality)} & \forall \point{p},\point{q},\point{r}
\in \mathcal{G} &\ d(\point{p},\point{q}) \leq d(\point{p},\point{r}) +
d(\point{r},\point{q})\\ 
\end{eqnarray*}
\end{definition}

\begin{definition}[Norm]
Given a metric $(d, \mathcal{G}, \mathcal{R})$, $d$ is called a {\em
norm} on the module $(\mathcal{G}, \mathcal{R}, +, \cdot)$ if
it also satisfies the following property:
\begin{eqnarray*}
\textrm{(positive homogeneity) } & \forall \point{p}, \point{q} \in
\mathcal{G}, \forall \alpha \in \mathcal{R}, & d(\alpha \cdot
\point{p}, \alpha \cdot \point{q}) = |\alpha| \times d(\point{p},
\point{q})\\ 
\end{eqnarray*}
\end{definition}

\begin{definition}[Image]
An {\em image} is a function $I\colon \mathcal{S} \rightarrow
\mathcal{R} \cup \infty$, where $\mathcal{S}$ is a finite subset of
$\mathcal{G}$.
\end{definition}

\begin{definition}[Distance map]\label{def:dist-map}
Given a binary image $I$, i.e., $I(\point{p})\in\{0,1\}$,let $X=\{
\point{p} \in \mathcal{S}, I(\point{p}) = 1 \}$ be the foreground and
$\overline{X}=\{ \point{p} \in \mathcal{S}, I(\point{p}) = 0 \}$ be
the background.  
Given a distance $(d, \mathcal{G}, \mathcal{R} \cup \infty)$, 
the {\em distance map}, $DM_{X}$, of $I$ is a grey level image,
$DM_{X}(\point{p}) \in \mathcal{R}$,  where the value of each point of
the foreground corresponds to its shortest distance to the background,
i.e. 
\begin{displaymath}
DM_{X} \ \colon \ \left\{ 
\begin{array}{l} 
\mathcal{S} \longrightarrow \mathcal{R} \cup \infty \\ 
\point{p} \longmapsto d(\point{p}, \overline{X}) = \inf_{\point{q} \in
\overline{X}} d(\point{p}, \point{q})\\  
\end{array} \right.
\end{displaymath}
\end{definition}

Given a module (i.e. a set of vectors) $(\mathcal{G}, \mathcal{R})$
and  an affine space (i.e, a set of points in $I$)
$\widetilde{\mathcal{G}}$,  we say that $\widetilde{\mathcal{G}}$ is
equivalent to $\mathcal{G}$ in the the following sense:
\begin{displaymath}
\forall \point{p}, \point{q} \in \widetilde{\mathcal{G}} \textrm{, }
\exists \vect{x} \in \mathcal{G} \textrm{, } \point{q} = \point{p} +
\vect{x} \textrm{; }
\end{displaymath}
this vector $\vect{x}$ is denoted $\vect{pq} = \point{q} -
\point{p}$. For every module $\mathcal{G}$, there exists an equivalent
affine space $\widetilde{\mathcal{G}}$ and vice versa. 
These spaces always have the same dimension. 
Given a point $O$ of $\widetilde{\mathcal{G}}$,
$(\widetilde{\mathcal{G}}, O)$ is an affine space with an origin. 
An origin is required to build a basis of $\widetilde{\mathcal{G}}$ and
allows to define the following operations. 
\begin{displaymath}
  \forall \point{p}, \point{q} \in \widetilde{\mathcal{G}}, \forall
  \lambda \in \mathcal{R} \textrm{,  } \lambda \point{p} = O + \lambda
  \vect{Op} \textrm{,  } \point{p} + \point{q} = O + \vect{Op} +
  \vect{Oq}.
\end{displaymath}
In the following, we will consider $\widetilde{\mathcal{G}}$ and
$\mathcal{G}$ as the same set and denote $\point{p} = \vect{Op}$.

\subsection{Chamfer masks and weighted distances}
 
We now consider more particularly sub-modules of $(\mathbb{R}^{n},
\mathbb{R})$, (for example $(\mathbb{Q}^{n}, \mathbb{Q})$ and $(\mathbb{Z}^{n},
\mathbb{Z})$) equivalent to the affine space
$\mathbb{R}^{n}$, with the origin
$O(0, 0, ..., 0)$ and the canonical basis $\left\{ (1, 0, 0, ..., 0),
(0, 1, 0, ..., 0), ...,(0, 0, 0, ..., 1)  \right\}$.
For any vector or point $\point{x} \in \mathcal{G}$, we denote
$\point{x} = (x^{i})_{i=1..n}$ the decomposition of $\point{x}$ in
this canonical basis.

Among  numerous other methods, distance maps can be computed
by propagation of local distances using a chamfer mask.
The latter is defined as follows:
\begin{definition}[Chamfer mask]
\label{def:chamfer-mask}
A {\em chamfer mask} $\mathcal{C}$ is a finite set of weighted vectors 
$\{ (\vect{v}_{k}, w_{k})_{k \in [1..m]} \in \mathcal{G} \times
\mathcal{R} \}$ which contains a basis of $(\mathcal{G}, \mathcal{R})$
and satisfies the following properties: 
\begin{eqnarray*}
(\textrm{positive weights}) & \forall k 
& w_k \in \mathbb{R}_+ \textrm{ and } \vect{v}_{k} \neq 0 \\
(\textrm{symmetry}) & (\vect{v}, w) \in \mathcal{C} & \Longrightarrow
( - \vect{v} , w ) \in \mathcal{C} 
\end{eqnarray*}  
\end{definition}

Intuitively, we speak about a distance between two points as the
length of the shortest path between these points. In the case of
weighted distances, we restrict the possible paths to those allowed by
the chamfer mask. 
In this case, a path is an ordered sequence of chamfer mask vectors
$\left\{ \vect{v}_{i_{1}}, \vect{v}_{i_{2}}, ... , \vect{v}_{i_{l}} \right\}$
with $\vect{v}_{i_{1}}$ having its origin at $\point{p}$, $\vect{v}_{i_{l}}$
having its end at $\point{q}$, and $i_{1}, i_{2}, ..., i_{l} \in [1...m]$. In
other words, $\forall j \in [1..l] \vect{v}_{i_{j}} \in \mathcal{C}$ and
$\vect{pq} = \vect{v}_{i_{1}} + \vect{v}_{i_{2}} + ... + \vect{v}_{i_{l}}$. 
As $\mathcal{G}$ is an Abelian group and a module on
$\mathcal{R}$, the order of the vectors does not matter, and we can
bring together the different vectors $(\vect{v}_{i_{j}})_{j \in [1..l]}$
of the path equal to the same mask vector $\vect{v}_{k}, k \in
[1..m]$. Observe that $l$ is the number of vectors within the path and
$m$ is the size of the chamfer mask. 
Moreover, as for each vector $\vect{v}_{k} \in
{\mathcal{C}}$, $- \vect{v}_{k} \in \mathcal{C}$, we can consider only
positive coefficients. We obtain the following definition:
\begin{definition}[Path from $\point{p}$ to $\point{q}$]
\label{def:chamfer-path}
Given a chamfer mask\\ 
$\mathcal{C}=\{ (\vect{v}_{k}, w_{k})_{k \in [1..m]} \in
\mathcal{G} \times \mathcal{R}\}$ and two points
$\point{p}, \point{q} \in \mathcal{G}$, a path
$\mathcal{P}_{\point{pq}}$ from $\point{p}$ to $\point{q}$ is a
sequence of vectors $\vect{v}_{k}$ of the mask $\mathcal{C}$ such
that: 
\begin{displaymath}
\mathcal{P}_{\point{pq}} = \sum_{k=1}^{m} \alpha_{k} \vect{v}_{k} = \vect{pq}
\textrm{ with } \forall k \in [1..m] \textrm{, } \alpha_{k} \in \mathcal{R}^{+}
\end{displaymath}
\end{definition}

\begin{definition}[Cost of a path]
\label{def:chamfer-cost}
The cost $\mathcal{W}$ of a such a path $\mathcal{P}_{\point{pq}}$ is defined by: 
\begin{displaymath}
\mathcal{W}(\mathcal{P}_{\point{pq}}) = \sum_{k=1}^{m} \alpha_{k} \cdot w_{k}
\end{displaymath}
\end{definition}

Since a mask $\mathcal{C}$ contains a basis of $\mathcal{G}$, and is
symmetric, such a path always exists for any couple of points
$(\point{p}, \point{q})$ with positive coefficients.

\begin{definition}[Weighted distance]
\label{def:weighted-distance}
A weighted distance $d_{\mathcal{C}}$ associated with a chamfer mask
$\mathcal{C}$ between two points $\point{p}$ and $\point{q}$ in
$\mathcal{G}$ is the minimum of the costs associated with paths
$\mathcal{P}_{\point{pq}}$ linking $\point{p}$ to $\point{q}$.
\begin{displaymath}
d_{\mathcal{C}}(\point{p}, \point{q}) = \min \left\{
\mathcal{W}(\mathcal{P}_{\point{pq}}) \right\} 
\end{displaymath}
\end{definition}

\subsection{Weighted distances properties \label{sec:wdp}}
\begin{theorem}
A weighted distance is invariant under translation.
\end{theorem}
\begin{proof}
  This proof can be found in \cite{thiel:hdr:2001}.
  Given any three points $\point{p}, \point{q},\point{r}$ in
  $\widetilde{\mathcal{G}}$, $d_{\mathcal{C}}(\point{r} + \point{p},
  \point{r} + \point{q}) = d_{\mathcal{C}}(\point{p}, \point{q})$.
  Indeed, $\vect{(r+p)(r+q)} = \vect{pq}$ and any path
  $\mathcal{P}_{\scriptscriptstyle{\vect{(r+p)(r+q)}}} = \sum_{k=1}^{m} \alpha_{k}
  \vect{v}_{k} = \vect{(r+p)(r+q)} = \vect{pq}$ from $\point{r+p}$ to
  $\point{r+q}$ is also a path from $\point{p}$ to $\point{q}$, with
  the same cost: $\mathcal{W}(\mathcal{P}_{\vect{(r+p)(r+q)}}) =
  \mathcal{W}(\mathcal{P}_{\vect{pq}}) = \sum_{k=1}^{m} \alpha_{k}
  w_{k}$. 
\end{proof}
\begin{corollary}
A property of a weighted distance between any two points
$(\point{q}, \point{r})$ of $\mathcal{G}$ can be expressed as a
property of a weighted distance between the origin $O$ and a point
$\point{p}$. 
\end{corollary}

Taking $\point{p} = \point{r} - \point{q}$ leads to 
$d_{\mathcal{C}}(O, \point{p}) = d_{\mathcal{C}}(O + \point{q},
\point{r}-\point{q}+\point{q}) = d_{\mathcal{C}}(\point{q}, \point{r})$.

In the following, we will express the properties of weighted distances
from the origin to any point $\point{p}$ of $\mathcal{G}$ and 
denote $d_{\mathcal{C}}(\point{p}) = d_{\mathcal{C}}(\point{O},\point{p})$.

To be able to forecast the weighted distance map, from a single point,
from the chamfer mask $\mathcal{C}$ , we divide $\mathcal{C}$ into several
sectors spanning the points of $\mathcal{G}$ such that the distance to
the origin of a point lying within a sector will only depend on a
restricted number (actually $n$ -- the dimension) of mask weights. We
define the following objects:
 
\begin{definition}[$\mathbb{R}$-sector]
Given a family of $n$ independent vectors of $\mathcal{G}$,
$\left(\vect{v}_{k} \right)_{k \in [1..n]}$, the {\em $\mathbb{R}$-sector}
$\langle \vect{v}_{1}, \vect{v}_{2}, ..., \vect{v}_{n} \rangle$ is
the region of $\mathbb{R}^{n}$ spanned by the vectors $\vect{v}_{1},
\vect{v}_{2}, ..., \vect{v}_{n}$ i.e.:
\begin{displaymath}
  \langle \vect{v}_{1}, \vect{v}_{2}, ..., \vect{v}_{n} \rangle =
  \left\{ \point{p} \in \mathbb{R}^{n} : \vect{p} = \sum_{k=1}^{n}
  \lambda_{k} \vect{v}_{k} \textrm{, } \lambda_{k} \in \mathbb{R}^{+} \right\} 
\end{displaymath}
\end{definition}

\begin{definition}[$\mathcal{G}$-sector]
The {\em $\mathcal{G}$-sector} $\llangle \vect{v}_{1}, \vect{v}_{2},
..., \vect{v}_{n} \rrangle$ is the set of points belonging to
$\mathcal{G}$ which are included in the $\mathbb{R}$-sector $\langle
\vect{v}_{1}, \vect{v}_{2}, ..., \vect{v}_{n} \rangle$. 
\end{definition}

\begin{figure}[ht!]
\begin{center}
\includegraphics[width=0.4\linewidth]{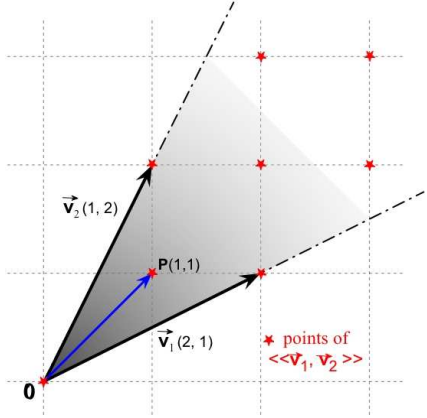}
\caption{\label{fig:R-G-sectors}Example of $\mathbb{R}$-sector and
$\mathcal{G}$-sector in $\left(\mathbb{Z}^{2}, \mathbb{Z}, +, \cdot
\right)$.}
\end{center}
\end{figure}

\begin{remark}
It is possible that a $\mathcal{G}$-sector does not correspond to the
set of points $\point{p}$ such 
that $\point{p} = \sum_{k=1}^{n} \lambda_{k} \vect{v}_{k} \textrm{, }
\lambda_{k} \in \mathcal{R}^{+}$. For example, as illustrated in
Figure \ref{fig:R-G-sectors}, if we take the module
$\left(\mathbb{Z}^{2}, \mathbb{Z}, +, \cdot \right)$, the point
$\point{p} {\scriptstyle \vect{(1,1)}}$ lies in the $\mathbb{R}$-sector
$\left\langle \vect{v}_{1} = {\scriptstyle \vect{(2,1)}}, \vect{v}_{2}
= {\scriptstyle \vect{(1,2)}} \right\rangle$ 
$\left(\vect{(1,1)} = 1/3 \cdot \vect{(2,1)} + 1/3 \cdot
\vect{(1,2)} \right)$ and as it is a point of $\mathbb{Z}^{2}$, it also
lies in the $\mathcal{G}$-sector $\llangle \vect{v}_{1}, \vect{v}_{2}
\rrangle$. But it cannot be written as $\lambda_{1} \cdot \vect{v}_{1}
+ \lambda_{2} \cdot \vect{v}_{2}$ with $\lambda_{1}, \lambda_{2} \in
  \mathbb{Z}$.
\label{remark:R-G-sectors} 
\end{remark}

\begin{definition}[Wedge of a chamfer mask]
We call a {\em wedge} of a chamfer mask $\mathcal{C}$, a
$\mathcal{G}$-sector formed by a family $\mathcal{F} = \left(
\vect{v}_{i_{k}} \right)_{k=1..n}$ of $n$ vectors of $\mathcal{C}$ which 
does not contain any other vector of $\mathcal{C}$.
\end{definition}
To avoid the situation illustrated in Remark 2.2 
we consider only $\mathcal{G}$-sectors based on a basis of
$\mathcal{G}$. 
\begin{definition}[$\mathcal{G}$-basis-sector]
We call a $\mathcal{G}$-sector $\llangle
\vect{v}_{1}, ... \vect{v}_{n} \rrangle$ where\\ $\left( \vect{v}_{k}
\right)_{k=1..n}$ is a basis of $\left(\mathcal{G}, \mathcal{R}, +,
\cdot \right)$ a $\mathcal{G}$-basis-sector.
\end{definition}
By definition of a basis, a $\mathcal{G}$-basis-sector corresponds
exactly to the the set of points $\point{p}$ such that $\vect{p} =
\sum_{k=1}^{n} \lambda_{k} \vect{v}_{k}, \lambda_{k} \in
\mathcal{R}^{+}$.

Given a family $\mathcal{F} = \left(\vect{v}_{k}\right)_{k \in
[1..n]}$ of $n$ independent vectors, we denote
$\Delta_{\mathcal{F}}^{0} \in \mathcal{R}$
\begin{displaymath}
\Delta_{\mathcal{F}}^{0} = 
\det (\vect{v}_{1}, \vect{v}_{2}, ..., \vect{v}_{n}) = 
\left| \begin{array}{l l l l}  
v^{1}_{1} & v^{1}_{2} & \cdots & v^{1}_{n} \\ 
v^{2}_{1} & v^{2}_{2} & \cdots & v^{2}_{n} \\ 
\vdots    & \vdots    & \ddots & \vdots    \\
v^{n}_{1} & v^{n}_{2} & \cdots & v^{n}_{n} \\ 
\end{array} \right| 
\end{displaymath}

and $\forall k \in [1..n]$, we consider the function
$\Delta_{\mathcal{F}}^{k} \ \colon\ \left\{ \begin{array}{l} 
\mathcal{G} \longrightarrow \mathcal{R} \\ 
\vect{x} \longmapsto \Delta_{\mathcal{F}}^{k}(\vect{x})\\
\end{array} \right.$ such that: 

\begin{displaymath}
\Delta_{\mathcal{F}}^{k}(\vect{x}) =  
\det (\vect{v}_{1}, ..., \vect{v}_{k-1}, \vect{x}, \vect{v}_{k+1},
..., \vect{v}_{n}) =  
\left| \begin{array}{c c c c c c c}  
v^{1}_{1} & \cdots & v_{k-1}^{1} & x^{1} & v_{k+1}^{1} & \cdots & v_{n}^{1} \\ 
v^{2}_{1} & \cdots & v_{k-1}^{2} & x^{2} & v_{k+1}^{2} & \cdots & v_{n}^{2} \\ 
\vdots    & \ddots & \vdots      & \vdots & \vdots     & \ddots & \vdots    \\
v^{n}_{1} & \cdots & v_{k-1}^{n} & x^{n} & v_{k+1}^{n} & \cdots & v_{n}^{n} \\ 
\end{array} \right|. 
\end{displaymath}

\begin{lemma}
\label{lemma:value}
A family $\mathcal{F} = \left( \vect{v}_{k} \right)_{k \in [1..n]}$
is a basis of $\left(\mathcal{G}, \mathcal{R}, +, \cdot \right)$, 
iff
\begin{displaymath}
\forall \vect{x} \in \mathcal{G} \textrm{, } \forall k \in [1..n]
\textrm{  }
\frac{1}{\Delta_{\mathcal{F}}^{0}} \times
\Delta_{\mathcal{F}}^{k}(\vect{x}) \in \mathcal{R}. 
\end{displaymath}
\end{lemma}

\begin{proof}
The proof in $\left(\mathbb{Z}^{2}, \mathbb{Z} \right)$ can be found
in \cite{hardy:book:1978}.\\
By definition of a basis, a family $\mathcal{F} = \left( \vect{v}_{k}
\right)_{k \in [1..n]}$ is a basis of $\left(\mathcal{G}, \mathcal{R},
+, \cdot \right)$ iff $\forall \point{x} \in \mathcal{G}$, $\exists
\alpha_{1}, \alpha_{2}, ..., \alpha_{n} \in \mathcal{R}$ such that
$\point{x} = \sum_{k=1}^{n} \alpha_{k} \cdot \vect{v}_{k}$. \\
As we consider sub-modules of $\mathbb{R}^{n}$,
$\point{x} \in \mathbb{R}^{n}$ and $\forall k \in [1..n]$, $\vect{v}_{k} \in
\mathbb{R}^{n}$. 
If the vectors of $\mathcal{F}$ are not independent, $\mathcal{F}$ is
not a basis of $\mathcal{G}$, and as $\Delta_{\mathcal{F}}^{0} = 0$,
$\forall \vect{x} \in \mathcal{G}, \frac{1}{\Delta_{\mathcal{F}}^{0}}
\Delta_{\mathcal{F}}^{k}(\vect{x}) \not\in \mathcal{R}$.
If $\mathcal{F}$ is an independent family of vectors, 
as $\mathbb{R}^{n}$ is a vector space, $\mathcal{F}$
is a basis of $\left( \mathbb{R}^{n}, \mathbb{R} \right)$ and $\exists
\alpha_{1}, \alpha_{2}, ..., \alpha_{n} \in \mathbb{R}$ such that
$\point{x} = \sum_{k=1}^{n} \alpha_{k} \cdot \vect{v}_{k}$.
This can be written: 
\begin{displaymath}
\begin{array}{c c c}
\left( \begin{array}{c} x^{1}\\ x^{2}\\ \vdots \\ x^{n}\end{array}\right) = 
&
\underbrace{
\left( \begin{array}{c c c c} 
v^{1}_{1} & v^{1}_{2} & \cdots & v^{1}_{n}\\
v^{2}_{1} & v^{2}_{2} & \cdots & v^{2}_{n}\\
\vdots    & \vdots    & \ddots & \vdots   \\
v^{n}_{1} & v^{n}_{2} & \cdots & v^{n}_{n}
\end{array}\right) } 
& 
\times \left( \begin{array}{c} \alpha_{1}\\ \alpha_{2}\\ \vdots \\
\alpha_{n} \end{array} \right) \\
 & \mathcal{M}_{\mathcal{F}} & 
\end{array}
\end{displaymath}
As $\mathcal{F}$ is an independent family, $\Delta_{\mathcal{F}}^{0}
= \det(\mathcal{M}_{\mathcal{F}})
\neq 0$ and the matrix $\mathcal{M}_{\mathcal{F}}$ can be inverted
such that 
\begin{displaymath}
\left( \begin{array}{c} \alpha_{1}\\ \alpha_{2}\\ \vdots \\
\alpha_{n} \end{array} \right) = 
\left( \begin{array}{c c c c} 
v^{1}_{1} & v^{1}_{2} & \cdots & v^{1}_{n}\\
v^{2}_{1} & v^{2}_{2} & \cdots & v^{2}_{n}\\
\vdots    & \vdots    & \ddots & \vdots   \\
v^{n}_{1} & v^{n}_{2} & \cdots & v^{n}_{n}
\end{array}\right)^{-1} 
 \left( \begin{array}{c} x^{1}\\ x^{2}\\ \vdots \\ x^{n}\end{array}\right).
\end{displaymath}
Inverting $\mathcal{M}_{\mathcal{F}}$ using Cramer's rule leads to 
\begin{displaymath}
 \forall k \in [1..n] \textrm{, } \alpha_{k} = \frac{1}{\Delta_{\mathcal{F}}^{0}} \times \Delta_{\mathcal{F}}^{k},
\end{displaymath}
and $\mathcal{F}$ is a basis of $\left( \mathcal{G}, \mathcal{R}
\right)$ iff $\forall k \in [1..n]$, $\alpha_{k} \in \mathcal{R}$.
\end{proof}
When $\left(\mathcal{G}, \mathcal{R} \right) = \left( \mathbb{Z}^{n},
\mathbb{Z} \right)$ we obtain the condition $\Delta_{\mathcal{F}}^{0}
= \pm 1$ for cone regularity defined in \cite{remy:iwcia:2000}. 

We can always organize $\mathcal{C}$ as a set of wedges (taking
$n$-tuples of vectors that do not contain any other vectors).
To avoid the situation in Remark 2.2 
and to be able to
forecast the final weighted distance from the chamfer mask, we choose
masks whose wedges are all $\mathcal{G}$-basis. When $(\mathcal{G},
\mathcal{R}) = (\mathbb{Z}^{2}, \mathbb{Z})$, taking every couple
$(\vect{u}, \vect{v})$ of adjacent vectors of $\mathcal{C}$ in
clockwise (or counter-clockwise) order gives such an organization. If,
for each couple $\det(\vect{u}, \vect{v}) = \pm 1$, then, they are
a $\mathcal{G}$-basis (cf. \cite{nacken:jmiv:1994}). 
For example, in Figure \ref{fig:shortestPath} of
Remark \ref{rem:shortestPath}, 
$\left\{ 
{\scriptstyle \llangle \vect{v}_{1}, \vect{v}_{2} \rrangle} \right.$, 
${\scriptstyle \llangle \vect{v}_{2}, \vect{v}_{3} \rrangle}$, 
${\scriptstyle \llangle \vect{v}_{3}, -\vect{v}_{1} \rrangle}$, 
${\scriptstyle \llangle -\vect{v}_{1}, -\vect{v}_{2} \rrangle}$, 
${\scriptstyle \llangle -\vect{v}_{2}, -\vect{v}_{3} \rrangle}$, 
$\left. {\scriptstyle \llangle -\vect{v}_{3}, \vect{v}_{1} \rrangle} \right\}$
is an organization of $\mathcal{C}$ in $G$-basis wedges. For $n \geq
3$, this organization may be more complicated. Indeed several ways of
organizing $n+1$ independent vectors into two wedges may exist: for
example, if we take the vectors $\vect{v}_{1} = \vect{(1,0,0)}$,
$\vect{v}_{2} = \vect{(1,1,0)}$, $\vect{v}_{3} = \vect{(1,1,1)}$,
$\vect{v}_{4} = \vect{(1,0,1)}$ and their symmetric vectors, the wedges
$\llangle \vect{v}_{1}, \vect{v}_{2}, \vect{v}_{3} \rrangle$ and
$\llangle \vect{v}_{3}, \vect{v}_{4}, \vect{v}_{1} \rrangle$ are
a $\mathcal{G}$-basis, but we can also consider the wedges $\llangle
\vect{v}_{4}, \vect{v}_{1}, \vect{v}_{2} \rrangle$ and $\llangle
\vect{v}_{2}, \vect{v}_{3}, \vect{v}_{4}
\rrangle$. In \cite{fouard:ivc:2005}, an automatic recursive method is
given; it is based on Farey series to organize chamfer masks of
$(\mathbb{Z}^{3}, \mathbb{Z})$ with $\mathcal{G}$-basis wedges. 
In the general case, considering a mask $\mathcal{C}$ containing only a
$\mathcal{G}$-basis $\mathcal{F} = (\vect{v}_{k}, w_{k})_{k=1..n}$ and
their symmetric wedges $\mathcal{F}_{0} = (-\vect{v}_{k}, w_{k})_{k=1..n}$,
leads to two
$\mathcal{G}$-basis wedges (if $\forall \vect{x} \in \mathcal{G},
\forall k \in [1..n], \frac{1}{\Delta_{\mathcal{F}}^{0}}
\Delta_{\mathcal{F}}^{k}(\vect{x}) \in \mathcal{R}$,  
$\frac{1}{\Delta_{\mathcal{F}_{0}}^{0}}
\Delta_{\mathcal{F}_{0}}^{k}(\vect{x}) = \pm \frac{1}{\Delta_{\mathcal{F}}^{0}}
\Delta_{\mathcal{F}}^{k}(\vect{x})\in \mathcal{R}$ as $\mathcal{R}$ is a
group). Moreover, the other wedges are $\mathcal{F}_{l} =
(\vect{v}_{1}, ..., \vect{v}_{l-1}, -\vect{v}_{l}, \vect{v}_{l+1},
..., \vect{v}_{n})$ for $l \in [1..n]$. They all are
$\mathcal{G}$-basis wedges as $\forall \vect{x} \in \mathcal{G},
\forall l \in [1..n], \frac{1}{\Delta_{\mathcal{F}_{l}}^{0}} 
\Delta_{\mathcal{F}_{l}}^{k}(x) = \pm \frac{1}{\Delta_{\mathcal{F}}^{0}}
\Delta_{\mathcal{F}}^{k}(x)\in \mathcal{R}$ as $\mathcal{R}$ is a
group. We can then add vectors to this mask, taking care of keeping an
organization in $\mathcal{G}$-basis wedges. 

\begin{remark}
\label{rem:shortestPath}
Given a chamfer mask $\mathcal{C} = \{ (\vect{v}_{k}, w_{k})_{k \in
[1..n]} \in \mathcal{G} \times \mathcal{R}\}$ which is organized in
$\mathcal{G}$-basis wedges, and given any point $\point{p}$ lying in a
wedge $\llangle \vect{v}_{i_{1}}, ..., \vect{v}_{i_{n}} \rrangle$ of
$\mathcal{C}$, there exist a path  $\mathcal{P}_{\point{p}} = \sum_{k=1}^{n}
\alpha_{k} \vect{v}_{k}$ with $\forall k \in [1..n]$ (i.e. the point
$\point{p}$ can by reached by a linear combination of only $n$ vectors
among the $m$ vectors of the mask $\mathcal{C}$ with coefficients in
$\mathcal{R}^{+}$).
However, the final weighted distance at this point may not be a linear
combination of the $n$ weights corresponding to {\em these}  $n$ vectors in
the chamfer mask.

For example, in $\left(\mathbb{Z}^{2}, \mathbb{Z}\right)$, let us
consider the mask $\mathcal{C}$ containing the weighted vectors
{\scriptsize $\left(\vect{v}_{1}, w_{1} \right) = \left( \vect{(1,0)},
2\right)$, $\left(\vect{v}_{2}, w_{2} \right) = \left( \vect{(2,1)},
5\right)$, $\left(\vect{v}_{3}, w_{3} \right) = \left( \vect{(1,1)},
1\right)$} and their symmetric vectors. 
The families {\scriptsize $\mathcal{F}_{1} = \left( \vect{v}_{1}, \vect{v}_{2}
\right)$} and {\scriptsize $\mathcal{F}_{2} = \left( \vect{v}_{2}, \vect{v}_{3}
\right)$}, generating the wedges 
{\scriptsize $\mathcal{S}_{1} = \llangle \vect{v}_{1}, \vect{v}_{2}
\rrangle$}
 and
{\scriptsize $\mathcal{S}_{2} = \llangle \vect{v}_{2}, \vect{v}_{3}
\rrangle$} respectively, are a basis of $(\mathbb{Z}^{2}, \mathbb{Z})$. 
Indeed,
$\Delta_{\mathcal{F}_{1}}^{0} = \left| \begin{array}{c c} 1 & 2\\ 0 & 1\\
\end{array}\right| = 1$ and $\forall \vect{x} \in \mathbb{Z}^2$,
$x^{1}$, $x^{2} \in \mathbb{Z}$ and: 
\begin{displaymath}
\frac{1}{\Delta_{\mathcal{F}_{1}}^{0}} \times
\Delta_{\mathcal{F}_{1}}^{1}(\vect{x}) = \frac{1}{1} \times \left|
\begin{array}{c c} x^{1} & 2\\ x^{2} & 1\\ \end{array}\right| = 1 \cdot x^{1}
- 2 \cdot x^{2} \in \mathbb{Z}
\end{displaymath}
\begin{displaymath}
\frac{1}{\Delta_{\mathcal{F}_{1}}^{0}} \times
\Delta_{\mathcal{F}_{1}}^{2}(\vect{x}) = \frac{1}{1} \times \left|
\begin{array}{c c}  1 & x^{1}\\ 0 & x^{2}\\ \end{array}\right| = 1 \cdot x^{2}
- 0 \cdot x^{1} \in \mathbb{Z}.
\end{displaymath}
In the same way for $\mathcal{F}_{2}$, we have:
$\Delta_{\mathcal{F}_{2}}^{0} = \left| \begin{array}{c c} 2 & 1\\ 1 & 1\\
\end{array}\right| = 1$ and $\forall \vect{x} \in \mathbb{Z}$,
$x^{1}$, $x^{2} \in \mathbb{Z}$ and: 
\begin{displaymath}
\frac{1}{\Delta_{\mathcal{F}_{2}}^{0}} \times
\Delta_{\mathcal{F}_{2}}^{1}(\vect{x}) = \frac{1}{1} \times \left|
\begin{array}{c c} x^{1} & 1\\ x^{2} & 1\\ \end{array}\right| = 1 \cdot x^{1}
- 1 \cdot x^{2} \in \mathbb{Z}
\end{displaymath}
\begin{displaymath}
\frac{1}{\Delta_{\mathcal{F}_{2}}^{0}} \times
\Delta_{\mathcal{F}_{2}}^{2}(\vect{x}) = \frac{1}{1} \times \left|
\begin{array}{c c}  2 & x^{1}\\ 1 & x^{2}\\ \end{array}\right| = 2 \cdot x^{2}
- 1 \cdot x^{1} \in \mathbb{Z}
\end{displaymath}

\begin{figure}[!ht]
\begin{center}
\includegraphics[width=0.8\linewidth]{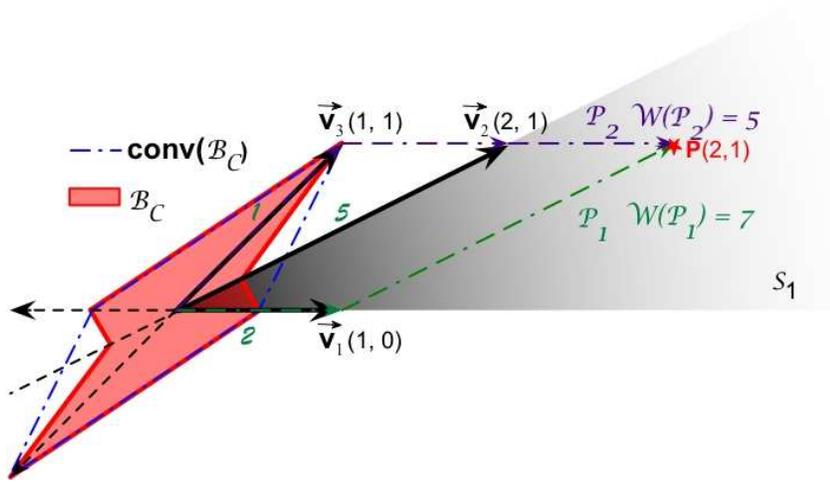}
\caption{\label{fig:shortestPath} Example where the weighted distance
does not depend on $\mathcal{G}$-basis sector vectors.}
\end{center}
\end{figure}
However, as illustrated in Figure \ref{fig:shortestPath} the point
$\point{p}= (2,1)$ lying in the 
$\mathcal{G}$-basis-sector $\mathcal{S}_{1}$ can be reached only by the
path $\mathcal{P}_{1} = 1 \cdot \vect{v}_{1} + 1 \cdot
\vect{v}_{2} = 1 \cdot \vect{v}_{2} + 1 \cdot \vect{v}_{1} $ containing
only vectors of $\mathcal{F}_{1}$ and $\mathcal{W}(\mathcal{P}_{1}) =
w_{1} + w_{2} = 2 + 5 = 7$. 
But $\point{p}$ can also be reached by the following path
$\mathcal{P}_{2} = 2 \cdot \vect{v}_{1} + 1 \cdot \vect{v}_{3} = 1
\cdot \vect{v}_{3} + 2 \cdot \vect{v}_{1} $ belonging neither to the
sector $\mathcal{S}_{1}$ nor to the sector $\mathcal{S}_{2}$ with a
cost $\mathcal{W}(\mathcal{P}_{2}) = 2 \times w_{1} + w_{3} = 2 \times 2
+ 1 = 5$.
As the weighted distance of a point (with respect to the origin) is
the minimum cost of all paths allowed by the mask, we have
$d_{\mathcal{C}}(\point{p}) \leq 5$ and thus
$d_{\mathcal{C}}(\point{p}) \neq \mathcal{W}(\mathcal{P}_{1})$.
\end{remark}

To avoid the situation mentioned in Remark \ref{rem:shortestPath}, we
add restrictions to the mask weights. These restrictions rely on the
fact that the polytope formed by the chamfer mask vectors normalized
by their weights is convex.

\begin{definition}[Normalized chamfer mask polytope] 
We call the polytope of $\mathbb{R}^{n}$ whose faces are the
$n-1$-dimensional pyramids formed by the $n$ vectors of each wedge of
a chamfer mask $\mathcal{M}_{C} = \{ (\vect{v}_{k}, w_{k})_{k \in
[1..m]} \in \mathcal{G} \times \mathcal{R}\}$
normalized by their weights, i.e. $\point{\widetilde{v}}_{i_{k}} =
\frac{1}{w_{i_{k}}} \cdot \point{v}_{i_{k}}\textrm{, } k \in [1..n]$
for each wedge $\llangle \vect{v}_{i_{1}}, ..., \vect{v}_{i_{n}}
\rrangle$ of $\mathcal{C}$ the {\em normalized chamfer mask polytope},
denoted $\mathcal{B}_{\mathcal{C}}$ .
Note that as $(\mathcal{G}, \mathcal{R}, +, \cdot)$ is a module but
not a vector space, these points may not be in $\mathcal{G}$.
\end{definition}

\begin{lemma}
\label{lemma:wedge:convex}
If the normalized polytope $\mathcal{B}_{\mathcal{C}}$ is convex, the
weighted distance of any point $\point{p}$ lying inside a wedge
$\llangle \vect{v}_{i_{1}}, ... \vect{v}_{i_{n}} \rrangle $ of a
chamfer mask $\mathcal{C} = \{ (\vect{v}_{k}, w_{k})_{k \in [1..n]}
\in \mathcal{G} \times \mathcal{R} \}$ depends only on the weights
$w_{i_{1}}, ..., w_{i_{n}}$.
\end{lemma}

\begin{proof}
  This proof can be found in \cite{nacken:jmiv:1994} for
  $(\mathbb{Z}^{2}, \mathbb{Z})$. \\
  If $\point{p} = \point{O}$, the proof is obvious.
  Let a point $\point{p} \neq \point{O}$, $\point{p} \in \llangle
  \vect{v}_{i_{1}}, ..., \vect{v}_{i_{n}} \rrangle$ be given. As $\mathcal{F} =
  (\vect{v}_{i_{1}}, ..., \vect{v}_{i_{n}})$ is a $\mathcal{G}$-basis,
  there exists a path $\mathcal{P} = \sum_{k=1}^{n} \alpha_{k}
  \vect{v}_{i_{k}}$ from $\point{O}$ to $\point{p}$ containing only
  vectors of $\mathcal{F}$ having positive coefficients ($\forall k
  \in [1..n], \alpha_{k} \geq 0$). The cost of this path is
  $\mathcal{W}(\mathcal{P}) = \sum_{k=1}^{n} \alpha_{k} \times w_{i_{k}}$. We
  can write $\mathcal{P}$ as 
  \begin{displaymath}
    \mathcal{P}  =  \left( \sum_{l=1}^{n}
      \alpha_{l} \times w_{i_{l}} \right) \times \sum_{k=1}^{n} \frac{\alpha_{k}
      \times w_{i_{k}}}{\sum_{l=1}^{n} \alpha_{l} \times w_{i_{l}}} \times
    \frac{1}{w_{i_{k}}} \cdot \vect{v}_{i_{k}}
      =  \mathcal{W}(\mathcal{P}) \cdot \vect{u}_{\mathcal{P}}
   \end{displaymath}
   with 
   \begin{displaymath}
     \vect{u}_{\mathcal{P}} = \sum_{k=1}^{n} \frac{\alpha_{k}
      \times w_{i_{k}}}{\sum_{l=1}^{n} \alpha_{l} \times w_{i_{l}}}
      \cdot \point{\widetilde{v}}_{i_{k}}.
   \end{displaymath}
   Since $\vect{u}_{\mathcal{P}}$ is a convex combination of the $n$
   normalized vectors
   of $\mathcal{F} \colon \forall k \in [1..n]$, $0 \leq  \frac{\alpha_{k}
     \times w_{i_{k}}}{\sum_{k=1}^{n} \alpha_{k} \times w_{i_{k}}}
   \leq 1$ and $\sum_{k=1}^{n} \frac{\alpha_{k}
     \times w_{i_{k}}}{\sum_{k=1}^{n} \alpha_{k} \times w_{i_{k}}} =
   1$, $\vect{u}_{\mathcal{P}}$ lies on
   $\mathcal{B}_{\mathcal{C}}$. Moreover, as the faces of
   $\mathcal{B}_{\mathcal{C}}$ are convex ($n-1$-dimensional polytope
   with $n$ vertices) $\vect{u}_{\mathcal{P}}$ also lies on the face
   formed by the family 
   $\mathcal{F}$ of $\mathcal{B}_{\mathcal{C}}$, i.e. on the
   boundary of $\mathcal{B}_{\mathcal{C}}$.\\
   Let us now consider another path $\mathcal{Q} = \sum_{k=1}^{m}
   \beta_{k} \vect{v}_{k}$ from $\point{O}$ to $\point{p}$ containing
   arbitrary vectors of $\mathcal{C}$. As $\mathcal{C}$ is symmetric,
   we can take $\forall k \in [1..n], \beta_{k} \geq 0$ without loss of
   generality. Then $\mathcal{W}(\mathcal{Q}) \geq
   \mathcal{W}(\mathcal{P})$. Indeed, 
  \begin{displaymath}
    \mathcal{Q}  =  \left( \sum_{k=1}^{m}
      \beta_{k} w_{{k}} \right) \cdot \sum_{k=1}^{m} \frac{\beta_{k}
      \times w_{k}}{\sum_{k=1}^{m} \beta_{k} \times w_{k}} \times
    \frac{1}{w_{k}} \cdot \vect{v}_{k}
     =  \mathcal{W}(\mathcal{Q}) \cdot \vect{u}_{\mathcal{Q}}
   \end{displaymath}
   with 
   \begin{displaymath}
     \vect{u}_{\mathcal{Q}} = \sum_{k=1}^{m} \frac{\beta_{k}
      \times w_{k}}{\sum_{k=1}^{m} \beta_{k} \times w_{k}}
      \cdot \point{\widetilde{v}}_{k}.
   \end{displaymath}
   $\vect{u}_{\mathcal{Q}}$ is a convex combination of $m$ normalized
   vectors of $\mathcal{C}$, and lies thus within the convex polytope
   $\mathcal{B}_{\mathcal{C}}$. 
   Moreover, we have $\vect{Op} = \mathcal{P} =
   \mathcal{W}(\mathcal{P}) \cdot \vect{u}_{\mathcal{P}} =
   \mathcal{W}(\mathcal{Q}) \cdot \vect{u}_{\mathcal{Q}} =
   \mathcal{Q}$ with $\vect{u}_{\mathcal{P}}$ and
   $\vect{u}_{\mathcal{Q}}$ having the same direction ($\vect{Op}$) 
   with positive coefficients.
   $\mathcal{W}_{\mathcal{P}} \left\| \vect{u}_{\mathcal{P}} \right\|
   = \mathcal{W}_{\mathcal{Q}} \left\| \vect{u}_{\mathcal{Q}}
   \right\|$. As $\vect{u}_{\mathcal{Q}}$ lies within
   $\mathcal{B}_{\mathcal{C}}$ and $\vect{u}_{\mathcal{P}}$ lies on
   the boundary of $\mathcal{B}_{\mathcal{C}}$, $\left\|
   \vect{u}_{\mathcal{P}} \right\| \geq \left\| \vect{u}_{\mathcal{Q}}
   \right\| $ (as $\mathcal{B}_{\mathcal{C}}$ is centered in
   $\point{O}$, if $\left\|\vect{u}_{\mathcal{Q}} \right\| > \left\|
   \vect{u}_{\mathcal{P}} \right\|$, $\point{u}_{\mathcal{Q}}$ would
   be farther from the origin than a point of the boundary, and thus
   outside $\mathcal{B}_{\mathcal{C}}$)    
   thus $\mathcal{W}(\mathcal{P}) \leq
   \mathcal{W}(\mathcal{Q})$ and $d_{\mathcal{C}}(\point{p}) =
   \mathcal{W}(\mathcal{P})$. 
\end{proof}
Figure \ref{fig:convex:proof} shows an example of vectors
$\vect{u}_{\mathcal{P}}$ and $\vect{u}_{\mathcal{Q}}$ of a mask
$\mathcal{C} \in \left(\mathbb{Z}^{2}, \mathbb{Z} \right)$. 

\begin{figure}[ht!]
  \begin{center}
    \includegraphics[width=0.9\linewidth]{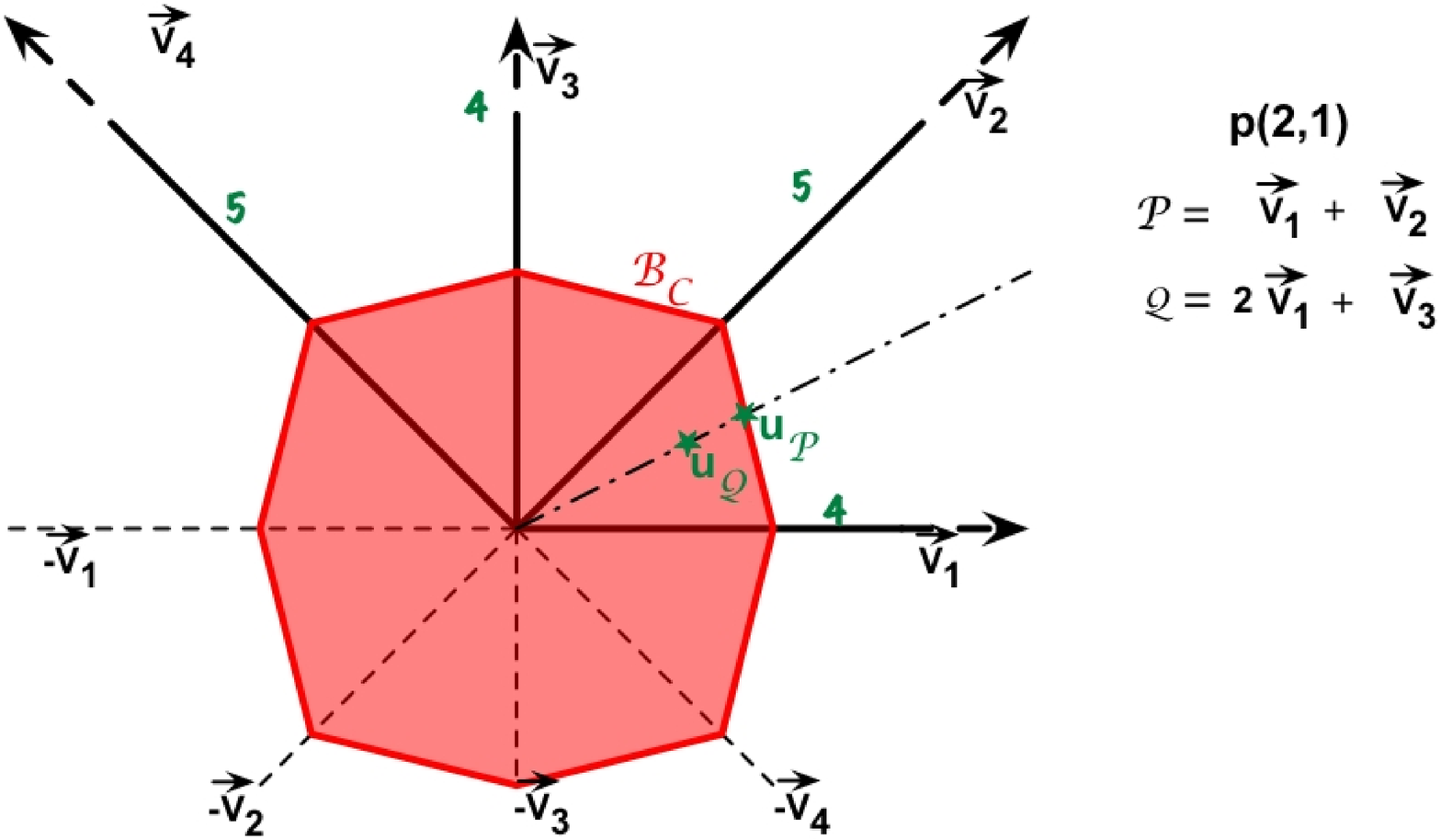}
\caption{\label{fig:convex:proof} Example where
$\mathcal{B}_{\mathcal{C}}$ is convex.   
We consider the vectors $\point{v}_{1}(1,0) \textrm{,  }
\point{v}_{2}(1,1) \textrm{, } \point{v}_{3}(0,1) \textrm{, }
\point{v}_{4}(-1,1)$ and the mask $\mathcal{C} = \left\{
(\point{v}_{1}, 4)\right.$, $(\point{v}_{2}, 5)$, $(\point{v}_{3}, 4)$,
$(-\point{v}_{1}, 4)$, $(-\point{v}_{2}, 5)$, $\left. (-\point{v}_{3}, 4)
\right\}$.
 We consider the path $\mathcal{P} = \point{v}_{1} +
   \point{v}_{2}$ and $\mathcal{Q} = 2 \cdot \point{v}_{1} +
   \point{v}_{3}$ from $\point{O}$ to $\point{p}(2, 1)$.
   We then have $ \| \point{u}_{\mathcal{Q}} \| \leq \|
   \point{u}_{\mathcal{P}} \|$.}
\end{center}
\end{figure}

\begin{corollary}
\label{corollary:convexHull}
If the vertices of each face of the convex hull
$\textrm{\textnormal{conv}}(\mathcal{B}_{\mathcal{C}})$ of
$\mathcal{B}_{\mathcal{C}}$ are normalized vectors corresponding to
$\mathcal{G}$-basis sectors of $\mathcal{C}$, then the vectors of
$\mathcal{C}$ whose corresponding normalized vectors do not lie on the
convex hull of $\mathcal{B}_{\mathcal{C}}$ are not used to compute the
final weighted distance. 
\end{corollary}
\begin{proof}
The proof can be found in \cite{remy:phd:2001,thiel:hdr:2001} for
vectorial spaces. 
Suppose there exists a vector $\vect{v}_{l} \in \mathcal{C}$ such
that $\vect{\widetilde{v}}_{l} = \frac{1}{w_{l}} \vect{v}_{l}$ does
not lie on the convex hull of $\mathcal{B}_{\mathcal{C}}$
(i.e. $\vect{\widetilde{v}}_{l}$ lies within
$\mathcal{B}_{\mathcal{C}}$). Let us
consider the $\mathcal{G}$-basis sector $\llangle \vect{v}_{i_{1}},
..., \vect{v}_{i_{n}} \rrangle$ in which $\vect{v}_{l}$ lies and whose
corresponding normalized vectors form a face $\mathcal{F}$ of
$\textrm{conv}(\mathcal{B}_{\mathcal{C}})$. There exists a path $\mathcal{P} =
\sum_{k=1}^{n} \alpha_{k} \vect{v}_{i_{k}} = \vect{v}_{l}$ from
$\point{O}$ to $\point{v}_{l}$ and with $\mathcal{W}(\mathcal{P})
< w_{l}$. Indeed, $\vect{v}_{l} = w_{l} \vect{\widetilde{v}}_{l} =
\mathcal{W}(\mathcal{P}) \vect{u}_{\mathcal{P}}$ with
$\vect{u}_{\mathcal{P}} = \sum_{k=1}^{n} \frac{\alpha_{k}
w_{i_{k}}}{\sum_{k=1}^{n} \alpha_{k} w_{i_{k}}}
\vect{\widetilde{v}}_{i_{k}}$ a convex combination of $n$ vectors of
$\mathcal{F}$ lying within $\mathcal{F}$. We thus have $\left\|
\vect{\widetilde{v}}_{l} \right\| < \left\| \vect{u}_{\mathcal{P}}
\right\|$ and $w_{l} > \mathcal{W}(\mathcal{P})$. Finally, as a
weighted distance is the minimum of the costs of all possible path, in
any path containing $\vect{v}_{l}$, $\vect{v}_{l}$ will be replaced by
$\mathcal{P}$ whose cost is lower.
Note that this is what happens in Remark \ref{rem:shortestPath}.
\end{proof}

\begin{remark}
\label{rem:convexHull}
If there exist faces of $\textrm{conv}(\mathcal{B}_{\mathcal{C}})$ formed by
normalized vectors whose corresponding mask vectors are not
$\mathcal{G}$-basis, the vectors of $\mathcal{C}$ which does not lie on
 $\textrm{conv}(\mathcal{B}_{\mathcal{C}})$ may be used, but this leads to
a final weighted distance which may not be homogeneous along some
directions, and we also may not be able to forecast the weighted
distance inside a wedge with the only vectors generating the wedge. 

For example, if we consider the mask $\mathcal{C}$  with the weighted
vectors: $\vect{v}_{1} = ({\scriptstyle \vect{(1,0)}}, 3)$,
$\vect{v}_{2} = ({\scriptstyle \vect{(1,1)}}, 2)$, $\vect{v}_{3} =
({\scriptstyle \vect{(0,1)}}, 3)$, $\vect{v}_{4} = ({\scriptstyle
\vect{(-1,1)}}, 2)$
and their symmetric vectors, as shown in Figure \ref{fig:nonConvex},
$d_{\mathcal{C}}(\point{O}, \point{(0, 2)}) = 4 \neq 2 \times
d_{\mathcal{C}}(\point{O}, \point{(0, 1)}) = 2 \times 3$. Moreover,
$d_{\mathcal{C}}(\point{O}, \point{(0,2)})$ does not only depend on
$\vect{v}_{2}$ and $\vect{v}_{3}$ or only on $\vect{v}_{3}$ and
$\vect{v}_{4}$.  
Note that the vectors $\vect{v}_{2}$ and $\vect{v}_{4}$
generating the corresponding face of
$\textrm{conv}(\mathcal{B}_{\mathcal{C}})$ do not form a
$\mathbb{Z}^{2}$ basis ($\det(\vect{v}_{1}, \vect{v}_{2}) = 2 \neq \pm
1$). 
\end{remark}

\begin{figure}[!ht]
\begin{center}
\includegraphics[width=0.7\linewidth]{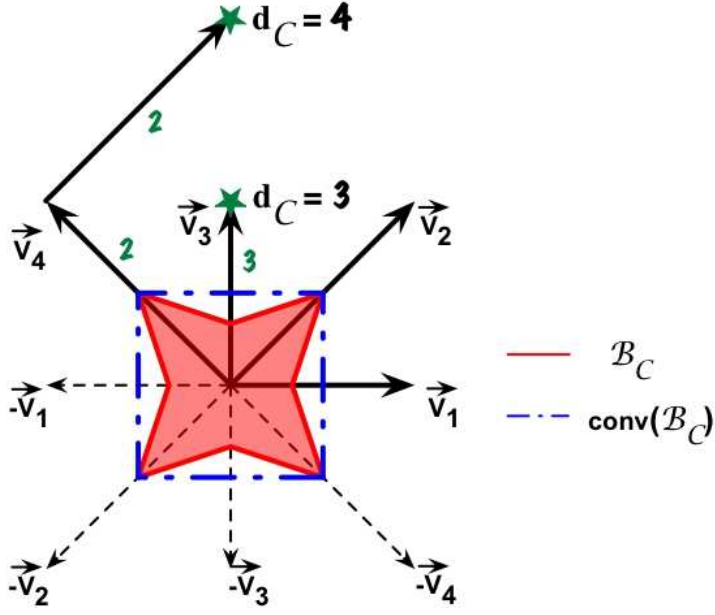}
\caption{\label{fig:nonConvex} When $\textrm{conv}(\mathcal{B}_{\mathcal{C}})
\neq \mathcal{B}_{\mathcal{C}}$ the weighted distance may not be
homogeneous.}
\end{center}
\end{figure}

In the following, we consider only chamfer masks whose normalized
polytope is convex. Indeed, if this is not the case, the mask may be
redundant (Corollary~\ref{corollary:convexHull}), or even if not, we
may not be able to forecast the final weighted distance, 
(Remark~\ref{rem:convexHull}). We can note that this condition implies
that collinear vectors (which are note opposite vectors) are removed
from the mask.

Considering previous lemmas and remarks, we can re-define a chamfer
mask with stronger conditions as follows: 

\begin{definition}[Chamfer mask (restricted)]
\label{def:chamfer-mask:restricted}
A {\em Chamfer Mask} $\mathcal{C}_\rho$ is a finite set of weighted vectors 
$\{ (\vect{v}_{k}, w_{k})_{k \in [1..m]} \in \mathcal{G} \times
\mathcal{R} \}$ which satisfies the following properties: 
\begin{eqnarray}
(\textrm{positive weights}) & \forall k \textrm{, } & w_k > 0 \textrm{ and } \vect{v}_{k} \neq 0 \label{maskCond:positiveWeights}\\
(\textrm{symmetry}) & (\vect{v}, w) \in \mathcal{C}_\rho & \Longrightarrow
( - \vect{v} , w ) \in \mathcal{C}_\rho \label{maskCond:symmetry}\\ 
(\textrm{Organized in }\mathcal{G} \textrm{-basis wedges}) & 
\left. \begin{array}{c}
  \forall \point{p} \in \mathcal{G}\\ 
  {\scriptstyle \exists \llangle \vect{v}_{i_{1}}, ...,
  \vect{v}_{i_{n}} \rrangle \in \mathcal{C}_\rho} \end{array} \right\}& 
\point{p} = \sum_{k=1}^{n} \alpha_{k} \vect{v}_{i_{k}} \label{maskCond:wedges}\\
(\textrm{Convex normalized polytope}) & & \mathcal{B}_{\mathcal{C}_\rho} =
\textrm{\textnormal{conv}}(\mathcal{B}_{\mathcal{C}_\rho}) \label{maskCond:convex}
\end{eqnarray}  
\end{definition}

\begin{theorem}
  \label{theorem:value}
  Given a chamfer mask $\mathcal{C}_\rho = \left\{ (\vect{v}_{k}, w_{k})_{k
  \in [1..n]} \in \mathcal{G} \times \mathcal{R} \right\}$, defined as
in Definition \ref{def:chamfer-mask:restricted}, the
  weighted distance of any point $\point{p}$ lying in a
  wedge $\llangle \vect{v}_{i_{1}}, ... \vect{v}_{i_{n}} \rrangle$
  can be expressed by:
  \begin{equation}
    d_{\mathcal{C}_\rho}(\point{p}) = \frac{1}{\Delta_{\mathcal{F}}^{0}}
    \times \sum_{k=1}^{n} \Delta_{\mathcal{F}}^{k}(\point{p}) \times w_{i_{k}}
  \end{equation}
\end{theorem}
\begin{proof}
This formula was given for $(\mathbb{Z}^{3}, \mathbb{Z})$ in
\cite{fouard:ivc:2005} without the entire proof.
Let $\point{p}$ be a point of $\mathcal{G}$ lying in the wedge
$W = \llangle \vect{v}_{i_{1}}, ... \vect{v}_{i_{n}} \rrangle$ of
$\mathcal{C}_\rho$. As $W$ is a basis of $(\mathcal{G}, \mathcal{R})$
(condition \ref{maskCond:wedges} of Definition
\ref{def:chamfer-mask:restricted}) there exists a path $\mathcal{P}$
from $\point{O}$ to $\point{p}$ and $\alpha_{1}, ..., \alpha_{k} \in
\mathcal{R}^{+}$ such that $\mathcal{P} = \vect{Op} = \sum_{k=1}^{m}
\alpha_{k} \vect{v}_{i_{k}}$. The proof of Lemma~\ref{lemma:value}
gives $\forall k \in [1..n], \alpha_{k} =
\frac{1}{\Delta_{\mathcal{F}}^{0}} \times
\Delta_{\mathcal{F}}^{k}(\point{p})$ and the proof of Lemma
\ref{lemma:wedge:convex} gives that the weight of $\mathcal{P}$ is
minimal. Thus 
$d_{\mathcal{C}_\rho}(\point{p}) = \mathcal{W}(\mathcal{P}) =
\frac{1}{\Delta_{\mathcal{F}}^{0}} \times \sum_{k=1}^{n}
\Delta_{\mathcal{F}}^{k}(\point{p}) \times w_{i_{k}}$. 
\end{proof}

\begin{theorem}
A weighted distance computed with a chamfer mask\\ 
$\mathcal{C}_\rho= \{ (\vect{v}_{k}, w_{k})_{k \in [1..m]} \in
\mathcal{G} \times \mathcal{R}\}$ as defined in Definition
\ref{def:chamfer-mask:restricted}  is a {\em norm} on $(\mathcal{G},
\mathcal{R})$. 
\end{theorem}

\begin{proof}
Here, we have to show that $d_{\mathcal{C}_\rho}$ is definite, positive and
symmetric and satisfies the triangular inequality and the positive
homogeneity properties. 
By definition of a weighted distance (Definition
\ref{def:weighted-distance}), given any points $\point{p} \in
\mathcal{G}$, there exist $\alpha_{1}$, $\alpha_{2}$,
..., $\alpha_{m}$ $\in \mathcal{R}$ such that
$d_{\mathcal{C}_\rho}(\point{O}, \point{p}) = d_{\mathcal{C}_\rho}(\point{p}) =
\sum_{k=1}^{n} \alpha_{k} w_{k}$ and $\mathcal{P}_{\point{p}} =
\vect{Op} = \point{p} = \sum_{k=1}^{m} \alpha_{k} \vect{v}_{k}$ (cf. 
Definitions \ref{def:chamfer-path}, \ref{def:chamfer-cost} and
\ref{def:weighted-distance}). 
\begin{enumerate}
\item {\em Positivity} (Needs conditions
  \ref{maskCond:positiveWeights} and \ref{maskCond:symmetry} of
  Definition \ref{def:chamfer-mask:restricted})\\
  By condition \ref{maskCond:symmetry} of Definition
  \ref{def:chamfer-mask:restricted} we can choose $\alpha_{k} \geq 0$
  for all $k \in [1..n]$. By Condition \ref{maskCond:positiveWeights}
  we also have $w_{k} \geq 0$ for all $k \in [1..n]$. We thus have
  $d_{\mathcal{C}_\rho}(\point{p}) \geq 0$.
  
\item {\em Definitivity} (Needs conditions
  \ref{maskCond:positiveWeights} and \ref{maskCond:symmetry} of
  Definition \ref{def:chamfer-mask:restricted})\\ 
  If $\point{p} = \point{O}$, then the path $\mathcal{P} =
  \sum_{k=0}^{m} 0 \cdot \vect{v}_{k}$ is a path from $\point{O}$ to
  $\point{p}$ and its cost is  $\mathcal{W}(\mathcal{P}) =
  \sum_{k=0}^{m} 0 \times w_{k} = 0$. Thus,
  $d_{\mathcal{C}_\rho}(\point{p}) = \min(\mathcal{W}(\vect{p})) \leq
  0$. By positivity, we also have $d_{\mathcal{C}_\rho}(\point{p})
  \geq 0$. It follows that $d_{\mathcal{C}_\rho}(\point{p}) = 0$.\\
  Conversely, if $d_{\mathcal{C}_\rho}(\point{p}) = 0$, as
  $\forall k \in [1..m]$, $w_{k} > 0$, we have $\forall k \in
  [1..m]$, $\alpha_{k} = 0$ and $\vect{p} = \sum_{k=1}^{m}
  0 \cdot \vect{v}_{k} = \vect{0}$. Thus $\point{p} = \point{O}$.

\item {\em Symmetry} (Needs the condition \ref{maskCond:symmetry} of
  Definition \ref{def:chamfer-mask:restricted})\\
  By condition \ref{maskCond:symmetry} of Definition
  \ref{def:chamfer-mask:restricted}, we have that 
  \begin{equation}\label{alpha-form}
\mathcal{P}_{\vect{pO}} = \sum_{k=1}^{m} \alpha_{k}
  (-\vect{v}_{k})
\end{equation}
 is a path from $\point{p}$ to $\point{O}$ and
  $\mathcal{W}(\mathcal{P}_{\vect{pO}}) =
  \mathcal{W}(\mathcal{P}_{\vect{Op}})$. Moreover,
  this cost is minimal. Indeed, let us consider a path
  $\mathcal{Q}_{\vect{pO}} = \sum_{k=1}^{m} \beta_{k} \vect{v}_{k}$ such
  that $\mathcal{W}(\mathcal{Q}_{\vect{pO}}) <
  \mathcal{W}(\mathcal{P}_{\vect{pO}})$. Then $\mathcal{Q}_{\vect{Op}}
  = \sum_{k=1}^{m} \beta_{k} \cdot (- \vect{v}_{k})$ is a path from
  $\point{O}$ to $\point{p}$ and $\mathcal{W}(\mathcal{Q}_{\vect{Op}}) <
  \mathcal{W}(\mathcal{P}_{\vect{Op}}) = d_{\mathcal{C}_\rho}(\point{p})$
  which is impossible by definition of a weighted
  distance (Definition \ref{def:weighted-distance}).
  Thus $d_{\mathcal{C}_\rho}(\point{p}, \point{O}) =
  \mathcal{W}(\mathcal{P}_{\vect{pO}}) =
  d_{\mathcal{C}_\rho}(\point{p})$.
  
\item {\em Triangular inequality} (Needs conditions
  \ref{maskCond:positiveWeights} and \ref{maskCond:symmetry} of
  Definition \ref{def:chamfer-mask:restricted})\\ 
  The proof can be found in \cite{verwer:phd:1991} for
  $(\mathbb{Z}^{2}, \mathbb{Z}, +, \cdot)$ and in
  \cite{kiselman:cviu:1996} (Corollary 3.4) for a general module. 
  We want $\forall \point{p}, \point{q} \in \mathcal{G}$,
  $d_{\mathcal{C}_\rho}(\point{p},\point{q}) \leq
  d_{\mathcal{C}_\rho}(\point{p}) + d_{\mathcal{C}_\rho}(\point{q})$. 
  Let $\mathcal{P}_{\point{q}} = \sum_{k=1}^{m} \beta_{k} \cdot
  \vect{v}_{k}$ be the minimum cost path between $\point{O}$ and
  $\point{q}$ (i.e. $d_{\mathcal{C}_\rho}(\point{q}) = \sum_{k=1}^{m} \beta_{k}
  w_{k}$), and by Eq.~(\ref{alpha-form}), $\mathcal{P}_{-\point{p}} =
  \sum_{k=1}^{m} \alpha_{k} 
  (-\vect{v}_{k})$. As the mask is symmetric,
  $\mathcal{P}_{-\point{p}}$ can be written $\mathcal{P}_{-\point{p}}
  = \sum_{k=1}^{m} \gamma_{k} \vect{v}_{k}$ for some $\gamma_k$ and
  $\mathcal{W}(\mathcal{P}_{-\point{p}}) =
  \mathcal{W}(\mathcal{P}_{\point{p}})$. 
  Let us suppose $d_{\mathcal{C}_\rho}(p,q) > d_{\mathcal{C}_\rho}(p) +
  d_{\mathcal{C}_\rho}(q)$. The path $\mathcal{P}_{\point{pq}} =
  \mathcal{P}_{-\point{p}} + \mathcal{P}_{\point{q}} = \sum_{k=1}^{m}
  (\gamma_{k} + \beta_{k}) \vect{v}_{k} = \vect{pO} + \vect{Oq} =
  \vect{pq}$ is a path from $\point{p}$ to $\point{q}$. As
  $(\mathcal{R}, +)$ is an Abelian group,
  $\mathcal{W}(\mathcal{P}_{\vect{pq}}) =
  \mathcal{W}(\mathcal{P}_{-\point{p}}) +
  \mathcal{W}(\mathcal{P}_{\point{q}}) = d_{\mathcal{C}_\rho}(\point{p}) +
  d_{\mathcal{C}_\rho}(\point{q}) < d_{\mathcal{C}_\rho}(\point{p}, \point{q})$
  which is impossible by definition of a weighted
  distance (Definition \ref{def:weighted-distance}). 
  By contradiction, we have
  $d_{\mathcal{C}_\rho}(\point{p}, \point{q}) \leq
  d_{\mathcal{C}_\rho}(\point{p}) + d_{\mathcal{C}_\rho}(\point{q})$.\\
\begin{remark}
Note that a decomposition in $\mathcal{G}$-basis wedges is not needed
for this condition. The only condition needed is to be able to extract
a basis among all mask vectors. This is the case, for example, for
masks containing only vectors corresponding to knight
displacements. Indeed, each wedge of this mask is not a
$\mathbb{Z}$-basis (see Figure~\ref{fig:R-G-sectors} and
Remark~2.2. 
However, this mask induces a metric,
\cite{das:PR:1988}. 
\end{remark} 
\begin{remark}
The triangular inequality does not depend on the choice of the
weights. 
\end{remark}

\item {\em Positive homogeneity} (Needs conditions
  \ref{maskCond:positiveWeights}, \ref{maskCond:symmetry},
  \ref{maskCond:wedges} and \ref{maskCond:convex} of
  Definition \ref{def:chamfer-mask:restricted})\\ 
  The proof can be found in \cite{remy:iwcia:2000} for
  $(\mathbb{Z}^{3}, \mathbb{Z}, +, \cdot)$ and in
  \cite{thiel:hdr:2001} for a general module.
  Let be $\lambda \in \mathcal{R}$.
  Let $W = \llangle \vect{v}_{i_{1}}, ..., \vect{v}_{i_{n}} \rrangle$
  be the wedges of $\mathcal{C}_\rho$ in which $\point{p}$ lies.
  By Theorem~\ref{theorem:value}, we have
  $d_{\mathcal{C}_\rho}(\point{p}) = \frac{1}{\Delta_{\mathcal{F}}^{0}}
  \times \sum_{k=1}^{n} \Delta_{\mathcal{F}}^{k}(\point{p}) \times
  w_{i_{k}}$.
  
  If $\lambda \geq 0$, the point $\lambda \cdot \point{p}$ also lies
  in the wedge $W$ (it is a point of $\mathcal{G}$, and $\lambda \point{p} =
  \sum_{k=1}^{n} \lambda \times \alpha_{k} \vect{v}_{i_{k}}$ with
  $\forall k \in [1..n], \lambda \times \alpha_{i_{k}} \in
  \mathcal{R}^{+}$). 
  By Theorem~\ref{theorem:value} we have: 
  \begin{eqnarray*}
    d_{\mathcal{C}_\rho}(\lambda \cdot \point{p}) & = &
    \frac{1}{\Delta_{\mathcal{F}}^{0}} \times \sum_{k=1}^{n}
    \Delta_{\mathcal{F}}^{k}(\lambda \cdot \point{p}) \times
    w_{i_{k}}\\
     & = & 
    \frac{1}{\Delta_{\mathcal{F}}^{0}} \times \sum_{k=1}^{n}
    \det \left(\vect{v}_{i_{1}}, ... \vect{v}_{i_{k-1}}, \lambda \cdot
    \vect{p}, \vect{v}_{i_{k+1}}, ... \vect{v}_{i_{n}} \right) \times
    w_{i_{k}} \\
     & = & 
    \frac{1}{\Delta_{\mathcal{F}}^{0}} \times \sum_{k=1}^{n} \lambda
    \times 
    \det \left(\vect{v}_{i_{1}}, ... \vect{v}_{i_{k-1}},
    \vect{p}, \vect{v}_{i_{k+1}}, ... \vect{v}_{i_{n}} \right) \times
    w_{i_{k}} \\
    & = & 
    \lambda \times \frac{1}{\Delta_{\mathcal{F}}^{0}} \times \sum_{k=1}^{n}
    \Delta_{\mathcal{F}}^{k}(\point{p}) \times w_{i_{k}} \\
    & = & 
    \lambda \times d_{\mathcal{C}_\rho}(\point{p})
  \end{eqnarray*}

  If $\lambda < 0$, $d_{\mathcal{C}_\rho}(\lambda \cdot \point{p}) =
  d_{\mathcal{C}_\rho}(-\lambda \cdot \point{p})$ with $-\lambda > 0$
  (symmetry property) and thus, $d_{\mathcal{C}_\rho}(\lambda \cdot
  \point{p}) = -\lambda \times d_{\mathcal{C}_\rho}(p) = |\lambda| \times
  d_{\mathcal{C}_\rho}(\point{p})$. $\blacksquare$
\end{enumerate}
\hspace{-\linewidth}
\end{proof}

\subsection{Weight optimization \label{sec:coeffsopt}}
 In addition to their metric and norm properties, weighted distances
 can be made more invariant to rotation.
 As a weighted distance is obtained by computing the smallest weight of
 several paths between two points, the first improvement to obtain a
 weighted distance with high rotational invariance is to
 allow a larger number of allowed directions for the paths. This means
 increasing precision by increasing the number of weighted vectors of
 the chamfer mask.

 To increase accuracy, another way is to choose suitable weights for
 mask vectors. This more challenging issue as been often addressed in
 the literature. The first optimal chamfer weights computation
 was performed for a 2-D $3\times3$ mask in a square grid
 \cite{borgefors:cvgip:1984}.
 Then authors computed optimal weights with different optimality criteria
 \cite{verwer:prl:1991,fouard:ivc:2005}, 
 for larger masks
 \cite{borgefors:cvgip:1986,verwer:prl:1991,svensson:cviu:2002} and
 for anisotropic grid
 \cite{coquin:prl:1998,mangin:espc:1994,sintorn:prl:2004}.  
 Authors of \cite{sintorn:prl:2004,fouard:ivc:2005,malandain:RR:2005} 
 proposed an automatic computation of
 optimal chamfer weights for rectangular grids.

 Observe that computing the distance transform and finding optimal
 weights is not directly related to the problem of estimating the
 length of straight lines in a discrete image
 \cite{bresenham:CGA:1996}. For optimal weights for such estimations,
 see \cite{dorst:PAMI:1986}. 

 In all of the previous papers, the computation of optimal chamfer weights
 is performed the same way:
 \begin{enumerate}
   \item First, a chamfer mask is built and decomposed in wedges. 

   \item Then, the final weighted distance from the origin to an
         arbitrary point of the grid is expressed. The variables
         corresponding to the mask weights are unknown, but
         variables corresponding to vector coordinates are known.
         In the general case, given a chamfer mask defined as
         Definition \ref{def:chamfer-mask:restricted}, and a point
         $\point{p} \in \mathcal{G}$ lying in the wedge $\mathcal{F} =
         \llangle \vect{v}_{i_{1}}, \vect{v}_{i_{2}}, ...
         \vect{v}_{i_{n}} \rrangle$ the value of the weighted distance
         at this point is 
          \begin{displaymath}
            d_{\mathcal{C}_\rho}(\point{p}) = \frac{1}{\Delta_{\mathcal{F}}^{0}}
            \times \sum_{k=1}^{n} \Delta_{\mathcal{F}}^{k}(\point{p})
            \times w_{i_{k}} 
          \end{displaymath}

   \item In the same way, the Euclidean distance from the origin to
         this point $\point{p}$ is expressed. In the general case,
         given a grid which is a sub-module of $\mathbb{R}^{n}$ with
         an elongation of $s_{1}, s_{2}, ..., s_{n}$ in each canonical
         direction, the Euclidean distance between the origin and the
         point $\point{p}$ can be expressed in the following way:
         \begin{displaymath}
           d_{E}(\point{p}) = \sqrt{\sum_{i=1}^{n} (s_{i} p^{i})^{2}}
         \end{displaymath}
         
   \item After these three steps, the error between the weighed
         distance and the Euclidean one can be expressed for any point
         $\point{p} \in \mathcal{G}$. This error can be either
         absolute (the difference between these two values)
         \cite{borgefors:cvgip:1986} or
         relative (the difference is divided by the Euclidean
         distance) \cite{verwer:prl:1991}.
         The error can be expressed as follows:
         \begin{eqnarray*}
           E(\point{p}) & = & d_{\mathcal{C}_\rho}(\point{p}) -
           d_{E}(\point{p})\\ 
           & = & \frac{1}{\Delta_{\mathcal{F}}^{0}} \times
           \left (\sum_{k=1}^{n}
             \Delta_{\mathcal{F}}^{k}(\point{p}) \times
             w_{i_{k}} \right) - \sqrt{\sum_{i=1}^{n} (s_{i}
             p^{i})^{2}}    
         \end{eqnarray*}
         The general relative error can be expressed as follows:
         \begin{eqnarray*}
           \label{eq:Erel}
           E_{rel}(\point{p}) & = & \frac{d_{\mathcal{C}_\rho}(\point{p}) -
           d_{E}(\point{p})}{d_{E}(\point{p})}\\ 
         & = & \frac{\sum_{k=1}^{n}
           \Delta_{\mathcal{F}}^{k}(\point{p}) \times
           w_{i_{k}}}{{\Delta_{\mathcal{F}}^{0}} \times
           \sqrt{\sum_{i=1}^{n} (s_{i} p^{i})^{2}}} - 1    
       \end{eqnarray*}
       
   \item The previous errors are $n$-dimensional functions of the
         coordinates $p^{1}, p^{2}$, ... $p^{n}$ of $\point{p}$. To
         reduce the number of these variables and be able to find
         extrema, the maximal error is computed either on a hyperplane
         or on a sphere. For example, in three dimensions, the error can be
         computed on a plane $X = T, Y = T \textrm{ or } Z = T$
         \cite{borgefors:cvgip:1986,fouard:ivc:2005} or on the sphere
         of radius $T$ \cite{malandain:RR:2005}. This error function
         is continuous on a compact set (the $n-1$ pyramid formed by
         the $n$ vectors of each wedge). Thus it is bounded and
         attains its bounds. These bounds can be located either at the
         vertices of the pyramid or inside (including other bounds
         such as edges) the pyramid.

   \item Computing the derivatives of the error function gives the
         point $\point{p_{max}}$ at which the error is maximum (due to
         the sign of the derivatives) inside
         the wedge. In this way, the maximum error within the wedge
         $E_{max} = E(\point{p_{max}})$ can be obtained (if $\point{p_{max}}$
         lies within the wedge).

   \item The other extrema (minima) $E_{1}, E_{2}, ..., E_{n}$ are obtained for
         the $n$ vectors delimiting the wedge (the vertices of the
         $n-1$ pyramid).  
         When computing the error on a sphere of radius $R$, in the
         general case, the extrema can be expressed in the following
         way ($\forall l \in [1..n]$): 
         \begin{displaymath}
           E^{l} = \left( w_{l} - ||\vect{v}_{l}|| \right)
           \textrm{ and }
           E_{rel}^{l} = \left( \frac{w_{l}}{||\vect{v}_{l}||} - 1 \right)
         \end{displaymath}
         with $||\vect{v}_{l}|| = \sqrt{\sum_{k=1}^{n} (s_{k}
         v_{l}^{k})^{2}}$ being the Euclidean norm of the vector
         $\vect{v}_{l}$ expressed in the world coordinate.
         
   \item Minimizing the maximum error leads to computing optimal real
         weights with the following equation: $E_{max} = - E_{1} =
         - E_{2} = ... = - E_{n}$.

   \item A depth-first search in an integer weights tree taking
         the error into account can lead to optimal integer weights set
         for a given mask \cite{fouard:ivc:2005,malandain:RR:2005}.
 \end{enumerate}

We have generalized weighted distance properties found in the
literature to modules. These properties are true for the well-known
cubic grid, but can also be applied to other grids such as the FCC and
BCC grids as we will see in Section~\ref{sec:bccfcc}. In applications,
an efficient algorithm for computing distance transforms is is
needed. In the following section, we discuss such an algorithm: the
two-scan chamfer algorithm. We prove that the algorithm produces
correct result for images on general point lattices.

\section{How to compute distance maps using the chamfer algorithm
\label{sec:chamferAlgo}} 

There are basically three families of algorithms for computing weighted
distance transforms -- bucket-sort (also known as wave-front
propagation), parallel, and sequential algorithms. Initially in the
bucket-sort algorithm \cite{dorst:ESPC:1986}, the border points of the
object are stored in a list. These 
points are updated with the distance to the background. The distances are
propagated by removing the updated points from the list and adding the
neighbors of these points to the list. This is iterated until the
list is empty. The
parallel algorithm \cite{rosenfeld:pr:1968} is the most intuitive
one. Given the original image, a new image is obtained by applying the
chamfer mask simultaneously to each point of the image giving the
minimum distance value at each of these points. This process is
applied to the new image. By applying the procedure iteratively until
stability, the distance map is obtained (a proof using infimal
convolution is found in \cite{kiselman:cviu:1996}). 

Rosenfeld \& Pfaltz \cite{rosenfeld:acm:1966} showed that this is
equivalent to two sequential scans of the image for grids in
$\mathbb{Z}^2$ and a $3 \times 3$ mask. Sequential means that we use
previously updated values of the same image to obtain new updated
values. Two advantages are that a two-scan algorithm is enough to
construct a distance map in any dimension and that the complexity is
known ($O (M)$, where 
$M$ is the number of points in  the image). The worst case complexity
of the parallel algorithm is $O(M^2)$ -- $M$ iterations of operations
on $M$ points. This result is now generalized and proved to be correct
on any grid for which all grid point coordinates can be expressed by a
basis.

In this section, $\mathcal{S}$ is a \textit{finite} subset of
$\mathcal{G}$ and $\mathcal{R}=\mathbb{Z}$.  
The cardinality of $\mathcal{S}$ is denoted $M=\textnormal{card}(\mathcal{S})$. 
Observe that the chamfer weights are in $\mathcal{R}^+=\mathbb{N}$
(Definition~\ref{def:chamfer-mask:restricted},
condition~\ref{maskCond:positiveWeights}) and that the weighted
distance have values in $\mathcal{R}^+$
(Definition~\ref{def:weighted-distance}). 

\begin{definition}
Given $\point{a}=(a_1,a_2,\dots,a_n)$, where $a_1,a_2,\dots,a_n\in
\mathcal{R}$ and $\sigma \in \mathbb{R}$. Let 
$$\mathcal{T}^\point{a}_{\sigma,\square}=\left\{\point{p} \in \mathcal{G}:
a_1 p^1 + a_2 p^2 + \dots + a_n p^n \square \sigma \right\},$$ 
where $\square$ is one of the relations $<,\leq,=,\geq,$ or $>$.
\end{definition}

For example, $\mathcal{T}^\point{a}_{\sigma,=}$ is a hyperplane and
$\mathcal{T}^\point{a}_{\sigma,<}$, $\mathcal{T}^\point{a}_{\sigma,\leq}$,
$\mathcal{T}^\point{a}_{\sigma,>}$, and $\mathcal{T}^\point{a}_{\sigma,\geq}$
are half-spaces separated by the hyperplane
$\mathcal{T}^\point{a}_{\sigma,=}$.

\begin{definition}[Scanning masks]
\label{def:scanningMasks}
Let a chamfer mask $\mathcal{C}$ be given. Let also $a_1,a_2,\dots,a_n
\in \mathcal{R}$ defining the hyperplane $\mathcal{T}^\point{a}_{0,=}$
such that $\mathcal{T}^\point{a}_{0,=} \cap \mathcal{C} = \emptyset$
be given.

The {\em scanning masks with respect to $\mathcal{C}$} are defined as
$$\mathcal{C}_1=\left\{ \left(\vect{v}_k, w_k \right) \in
\mathcal{C}: \vect{v}_k \in \mathcal{T}^\point{a}_{0,<}\right\}$$ 
 $$\mathcal{C}_2=\left\{ \left(\vect{v}_k, w_k \right) \in
 \mathcal{C}: \vect{v}_k \in \mathcal{T}^\point{a}_{0,>}
 \right\}$$
\end{definition} 
In the following, the notation $\vect{v}_i \in \mathcal{C}_l$ and
$w_i \in \mathcal{C}_l$, $l \in \{1,2\}$, means that the pair
$\left(\vect{v}_i, w_i \right) \in \mathcal{C}_l$ will be used. 

\begin{definition}[Scanning order]
A {\em scanning order} is an ordering of the
$M=\textnormal{card}(\mathcal{S})$ points in $\mathcal{S}$, denoted
$\point{p}_1, \point{p}_2, \dots, \point{p}_M$.
\end{definition}
For a scanning mask to propagate distances efficiently, it is
important that, in each step of the propagation, the values at the
points in $\mathcal{S}$ from which the mask can propagate distances
have already been updated. This is guaranteed if each point that can
be reached by the scanning mask either has already been visited or is
outside the image.  

\begin{definition}[Mask supporting a scanning order]
Let $\point{p}_1, \point{p}_2, \dots, \point{p}_M$ be a scanning order
and $\mathcal{C}_l$ a scanning mask. The scanning mask $\mathcal{C}_l$
{\em supports the scanning order} if $$\forall \point{p}_i, \forall
\vect{v}_j \in \mathcal{C}_l, \exists
i'<i:\left(\point{p}_{i'}=\point{p}_i+\vect{v}_j \textrm{ or }
\point{p}_i+\vect{v}_j \notin \mathcal{S} \right).$$ 
\end{definition}

\begin{proposition}
Given a scanning mask $\mathcal{C}_l$ and an image $\mathcal{S}$,
there is a scanning order such that $\mathcal{C}_l$ supports the
scanning order. 
\end{proposition}

\begin{proof}
Let $\mathcal{T}^\point{a}_{0,=}$ such that
$\mathcal{T}^\point{a}_{0,=} \cap \mathcal{C} = \emptyset$ be the
hyperplane defining $\mathcal{C}_1$ and $\mathcal{C}_2$ (we consider
$\mathcal{C}_1$ here, the proof for $\mathcal{C}_2$ is similar). Now
we consider two sets, $\mathcal{S} \setminus \mathcal{V}$ and
$\overline{\mathcal{S}} \cup \mathcal{V}$, where
$\overline{\mathcal{S}}$ is the complement of $\mathcal{S}$ and the
elements in $\mathcal{V}$ are the already visited points in
$\mathcal{S}$. Let $\mathcal{V}= \{\point{p}_1,
\point{p}_2,\dots,\point{p}_N\}$,
$N<\textnormal{card}(\mathcal{S})$. For every $\point{p} \in
\mathcal{S}$, there is a $\sigma_{\point{p}}\in \mathbb{R}$
($\sigma_{\point{p}}=a_1 p^1 + a_2 p^2 + \dots + a_n p^n$) such that
$\point{p} \in \mathcal{T}^\point{a}_{\sigma_{\point{p}},=}$. Since
$\mathcal{S} \setminus \mathcal{V}$ is finite, there is a point
$\point{p} \in \mathcal{S} \setminus \mathcal{V}$ with corresponding
$\sigma_{\point{p}}$ such that $\point{p} = (\mathcal{S} \setminus
\mathcal{V}) \cap \mathcal{T}^\point{a}_{\sigma_{\point{p}},=}$,
$\mathcal{S} \setminus \mathcal{V}= (\mathcal{S} \setminus
\mathcal{V}) \cap \mathcal{T}^\point{a}_{\sigma_{\point{p}},\geq}$, and
$\emptyset= (\mathcal{S} \setminus \mathcal{V}) \cap
\mathcal{T}^\point{a}_{\sigma_{\point{p}},<}$. Thus, $\forall \vect{v} \in
\mathcal{C}_1, a_1 \cdot p^1 + a_2 \cdot p^2 + \dots + a_n\cdot
p^n=\sigma_{\point{p}}$ and $a_1 \cdot v^1 + a_2 \cdot v^2 + \dots +
a_n\cdot v^n<0$, so $a_1 \cdot (p^1+v^1) + a_2 \cdot (p^2+v^2) + \dots
+ a_n\cdot (p^n+v^n) < \sigma_{\point{p}}$ and thus $\point{p}+\vect{v} \in
\mathcal{T}^\point{a}_{\sigma_{\point{p}},<}$ which has an empty
intersection with $\mathcal{S} \setminus \mathcal{V}$, so
$\point{p}+\vect{v} \notin (\mathcal{S} \setminus \mathcal{V})$. Let
this $\point{p}$ be $\point{p}_{N+1}$, assign it to $\mathcal{V}$ and
repeat until $\mathcal{S}=\mathcal{V}$. 
\end{proof}

\begin{definition}[Chamfer algorithm]\label{chamfer_alg}
Let the scanning masks $\mathcal{C}_1$ and $\mathcal{C}_2$ 
and scanning orders $\point{p}_1, \point{p}_2, \dots, \point{p}_M$ and
$\point{p}_M, \point{p}_{M-1}, \dots, \point{p}_1$ such that the masks
support these scanning orders be given. 
\begin{itemize}
\item Initially, $\forall \point{p} \in X, f(\point{p})\leftarrow
      +\infty$ and  $\forall \point{p} \in \overline{X}, f(\point{p})
      \leftarrow 0$. 
\item The image is scanned two times using the two scanning orders.
\item For each visited point $\point{p}_i$ in scan $l$,
      $$f(\point{p}_i)\leftarrow \min\left(\min_{\left(\vect{v}_j,
      w_j\right)\in \mathcal{C}_l: \left(\point{p}_i+\vect{v}_j
      \right)\in \mathcal{S}} \left(w_j +
      f(\point{p}_i+\vect{v}_j) \right), f(\point{p}_i) \right).$$ 

\end{itemize}
\end{definition}

\begin{remark}
Usually, in the first scan $f(\point{p}_i)$ is omitted, as
$f(\point{p}_i)$ is only computed for $X$, so we know
$f(\point{}_i)=\infty$. 
\end{remark}

Now we prove that, in a distance map $DM_X$, for any point
$\point{p}'$ on a path defining the distance between $\point{p}\in X$
and its closest background point $\point{q} \in \overline{X}$, the
value in the distance map at point $\point{p}'$ is given by
$DM_X(\point{p}')=d(\point{q},\point{p}')$. In other words, the
closest point in the background from $\point{p}'$ is $\point{q}$. 

\begin{lemma}\label{chamfer_lemma}
Let $\point{p}\in X$ and $\point{q}\in \overline{X}$ be such that
$d(\point{p},\point{q})=d(\point{p}, \overline{X})$ and
$\mathcal{I}=\left\{ i:\left(\vect{v}_i,w_i\right)\in
\mathcal{C}\right\}$. There is a set $\{\alpha_i\in\mathcal{R}\}$ such
that the point $\point{p}$ can be written $\displaystyle
\point{p}=\point{q}+\sum_{i\in \mathcal{I}} \alpha_i \vect{v}_i$ and
$\displaystyle d(\point{p}, \point{q})=\sum_{i\in \mathcal{I}}
\alpha_i w_i$. 

For any set $\left\{ \beta_i \in \mathcal{R}: 0\leq\beta_i \leq
\alpha_i \right\}$, if $\point{p}'\displaystyle = \point{q}+\sum_{i\in
\mathcal{I}} \beta_i \vect{v}_i\in \mathcal{S}$, then $\displaystyle
d(\point{p}', \overline{X}) = d(\point{p}',\point{q})=\sum_{i\in
\mathcal{I}} \beta_i w_i$. 
\end{lemma}

\begin{proof}
Let $d\left( \point{p}',\overline{X} \right)=K \neq \sum \beta_i
w_i$. Since $\sum \beta_i w_i$ is the length of a path
between $\point{p}'=\point{q}+\sum \beta_i \vect{v}_i$ and
$\point{q}$, the weighted distance can not be larger than this. Assume
that $K<\sum \beta_i w_i$. Then there is a $\point{q}' \in
\overline{X}$ such that 
\begin{eqnarray}
K=d\left(\point{q}',\point{p}'\right) & < & d\left(\point{q},
\point{p}'\right) \textrm{ and}\nonumber \\ 
d\left(\point{q}',\point{p}\right) & \geq &
d\left(\point{q},\point{p}\right)=d\left(\point{p},\overline{X}\right)
\textrm{.}\nonumber 
\end{eqnarray}
Now,
\begin{displaymath}
d(\point{q},\point{p})=\sum \alpha_i w_i =\sum \beta_i w_i
+\sum (\alpha_i-\beta_i) w_i >d\left(\point{q}',\point{p}'\right)
+ \sum (\alpha_i-\beta_i) w_i. 
\end{displaymath}
Since $\sum (\alpha_i-\beta_i) w_i$ is the length of one path (we
cannot assume that it is a shortest path) between the points
$\point{p}'=\point{q}+\sum \beta_i \vect{v}_i$ and
$\point{p}=\point{q}+\sum \alpha_i \vect{v}_i$, it follows that $\sum
(\alpha_i-\beta_i) w_i \geq d\left(\point{p}',\point{p}\right)$
and thus 
\begin{displaymath}
d\left(\point{q}',\point{p}'\right) + \sum (\alpha_i-\beta_i)
w_i\geq d\left(\point{q}',\point{p}'\right) +
d\left(\point{p}',\point{p}\right) \geq d(\point{q}',\point{p}), 
\end{displaymath}
which contradicts $d\left(\point{q}',\point{p}\right) \geq
d\left(\point{q},\point{p}\right)$. 
\end{proof}

To assure that the propagation does not depend on points outside the
image, the definition of border points \textit{or} of wedge preserving
images defined below can be used. 

\begin{definition}[Border point]\label{def:borderpoint}
Given a chamfer mask $\mathcal{C}$, a point $\point{p} \in
\mathcal{S}$ is {\em a border point} if $$\exists \vect{v} \in
\mathcal{C}: \point{p}+\vect{v} \notin \mathcal{S}.$$ 
\end{definition}

\begin{lemma}\label{border_lemma}
If for all border points $\point{r} \in \mathcal{S}, \point{r} \in
\overline{X}$, then for all $\point{p}\in X ,\point{q} \in
\overline{X}$ such that
$d(\point{p},\overline{X})=d(\point{p},\point{q})$, all points in any
shortest path between $\point{p}$ and $\point{q}$ are in
$\mathcal{S}$. 
\end{lemma}

\begin{proof}
Assume that a point $\point{p}'$ in a shortest path between
$\point{p}$ and $\point{q}$ is not in $\mathcal{S}$. Then, since all
border points are in the background, there is a background grid point
in all paths between $\point{p}$ and $\point{p}'$. Since
$d(\point{p},\point{p}')<d(\point{p},\point{q})$, it follows that
there must be a border point $\point{p}'' \in \overline{X}$ such that
$d(\point{p},\point{p}'')<d(\point{p},\point{p}')<d(\point{p},\point{q})$
which contradicts the assumption $d(\point{p},\point{q})=d(\point{p},
\overline{X})$. 
\end{proof}

\begin{definition}[Wedge-preserving half-space]\label{def:wedge_pres}
Given a chamfer mask $\mathcal{C}$ and $\sigma \in \mathbb{R}$,
$a_1,a_2,\dots,a_n \in \mathcal{R}$ such that $\exists \point{p} \in
\mathcal{S} : \point{p} \in \mathcal{T}^\point{a}_{\sigma,=}$. The
half-spaces $\mathcal{T}^\point{a}_{\sigma,\geq}$ and
$\mathcal{T}^\point{a}_{\sigma,\leq}$ are {\em wedge-preserving } if
$$\forall \point{p} \in \mathcal{T}^\point{a}_{\sigma,=}, \forall
W,\textrm{ either } \left( \forall \vect{v} \in W, \point{p}+\vect{v}
\in \mathcal{T}^\point{a}_{\sigma,\leq}\right) \textrm{ or } \left( \forall
\vect{v} \in W, \point{p}+\vect{v} \in \mathcal{T}^\point{a}_{\sigma,\geq}
\right),$$ 
where $W$ denotes wedges in $\mathcal{C}$.
\end{definition}

\begin{definition}[Wedge-preserving image]\label{wedge_pres_halfspace}
Given a chamfer mask $\mathcal{C}$, the image $\mathcal{S}$ is {\em
wedge-preserving} if it is the intersection of wedge-preserving
half-spaces. 
\end{definition}

\begin{lemma}\label{lem_wedge_equival}
If an image $\mathcal{S}$ is wedge-preserving, then for all
$\point{p},\point{q} \in \mathcal{S}$, all points in any shortest path
between $\point{p}$ and $\point{q}$ are in $\mathcal{S}$. 
\end{lemma}

\begin{proof}
Let $\mathcal{S}$ be wedge-preserving and let $\point{p},\point{q} \in
\mathcal{S}$ such that $\displaystyle \point{p}= \point{q}+ \sum_{i\in
\mathcal{I}} \alpha_i \vect{v}_i$ with
$\mathcal{I}=\left\{i:\vect{v}_i \in W\right\}$, where $W$ is a
wedge. Assume that there is a fixed $k \in \mathcal{I}$ and a set
$\{\gamma_i:0 \leq \gamma_i \leq \alpha_i \textrm{ and }
\gamma_k<\alpha_k\}$ such that $\displaystyle \point{p}'= \point{q}+
\sum_{i\in \mathcal{I}} \gamma_i \vect{v}_i \in \mathcal{S}$ but
$\point{p}'+\vect{v}_k \notin \mathcal{S}$. In other words, if there
is a point not in $\mathcal{S}$ in the shortest path between
$\point{p}$ and $\point{q}$, then some point $\point{p}'$ in
$\mathcal{S}$ has a neighbor $\point{p}'+\vect{v}_k$ that is not in
$\mathcal{S}$. Since $\mathcal{S}$ is wedge-preserving, there is a
plane defined by $\sigma,a_1,a_2,\dots,a_n \in \mathcal{R}$ such that
$\point{p}'\in \mathcal{T}^\point{a}_{\sigma,\leq}$ and
$\point{p}'+\vect{v}_k\in \mathcal{T}^\point{a}_{\sigma,>}$. Since 
$$a_1 {p'}^1 + a_2 {p'}^2 + \dots + a_n {p'}^n \leq \sigma \textrm{ and}$$
$$a_1 ({p'}^1+v_k^1) + a_2 ({p'}^2+v_k^2) + \dots + a_n ({p'}^n+v_k^n)
> \sigma \textrm{, we get}$$ 
$$a_1 v_k^1 + a_2 v_k^2 + \dots + a_n v_k^n > 0.$$ 
Taking any $\point{r} \in \mathcal{T}^\point{a}_{\sigma,=}$,
$\point{r}+\vect{v}_k \in \mathcal{T}^\point{a}_{\sigma,>}$, which by
Definition~\ref{wedge_pres_halfspace} implies that $\forall \point{r}
\in \mathcal{T}^\point{a}_{\sigma,=}, \forall \vect{v} \in W,
\point{r}+\vect{v} \in \mathcal{T}^\point{a}_{\sigma,\geq}$. This implies
that  
\begin{equation}\label{geqvect}
\forall \vect{v} \in W, \vect{v} \in \mathcal{T}^\point{a}_{0,\geq}.
\end{equation}
Now, $\point{p}=\displaystyle \point{p}'+\sum (\alpha_i - \gamma_i)
\vect{v}_{i} \in \mathcal{S}$, but 
$$a_1 ({p'}^1+\sum (\alpha_i - \gamma_i) v_{i}^1) + a_2 ({p'}^2+\sum
(\alpha_i - \gamma_i) v_{i}^2) + \dots + a_n ({p'}^n+\sum (\alpha_i -
\gamma_i) v_{i}^n)=$$ 
$$=a_1 \left( {p'}^1+\vect{v}_k^1 \right) + a_2
\left({p'}^2+\vect{v}_k^2 \right) + \dots + a_n
\left({p'}^n+\vect{v}_k^n \right) +$$ 
$$+\sum_{i \in \mathcal{I}, i \neq k} (\alpha_i - \gamma_i)\left(a_1
v_i^1 + a_2 v_i^2 + \dots + a_n v_i^n \right)
+(\alpha_k-\gamma_k-1)\left(a_1 v_k^1 + a_2 v_k^2 + \dots + a_n v_k^n
\right) >$$ 
$$>\sigma+\sum_{i \in \mathcal{I}, i \neq k} (\alpha_i - \gamma_i)0+(\alpha_k-\gamma_k-1)0=\sigma,$$
since $\point{p}'+\vect{v}_k\in \mathcal{T}^\point{a}_{\sigma,>}$,
$\alpha_i \geq \gamma_i$ if $i\neq k$ and $\alpha_k \geq \gamma_k+1$,
and using Eq.~\ref{geqvect}. 
We have $\displaystyle \point{p} \in \mathcal{T}^\point{a}_{\sigma,>}$,
i.e. $\displaystyle \point{p} \notin \mathcal{S}$. Contradiction. 
\end{proof}

\begin{remark}
As the image generally is given first, a mask can often be built with
preserved wedges and then cut into two scanning masks. Then the points
in the image can be ordered such that the masks support the scanning
directions and then the distance map can be computed using our
algorithm.\ \ 
\end{remark}

\begin{figure}[ht!]\label{chamf-alg-conditions}
\begin{center}
\subfigure[]{\includegraphics[height=0.2\linewidth]{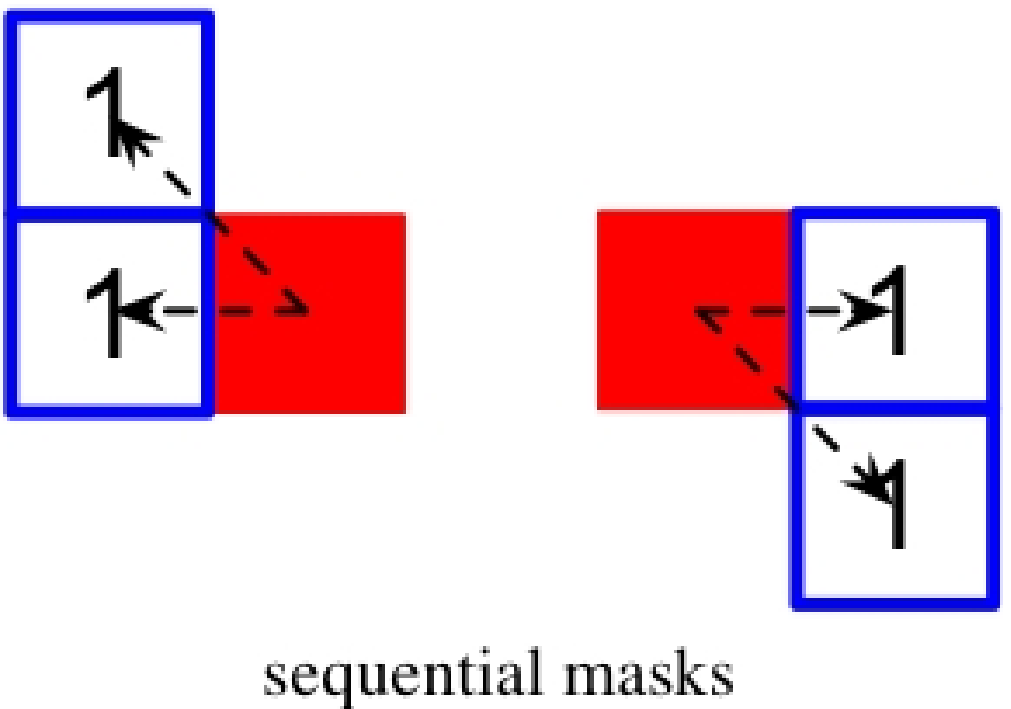}}
\subfigure[]{\includegraphics[height=0.28\linewidth]{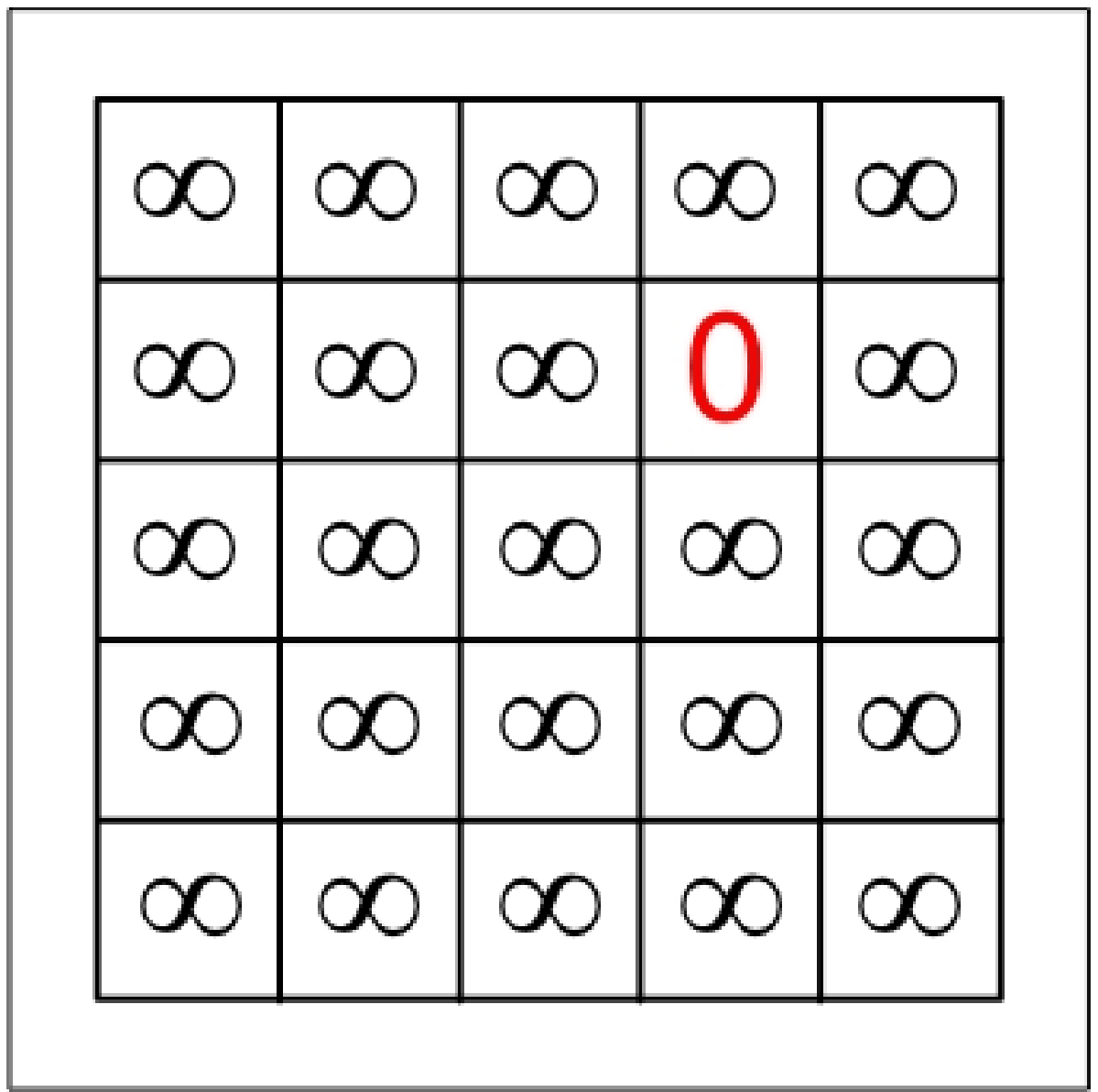}}
\subfigure[]{\includegraphics[height=0.28\linewidth]{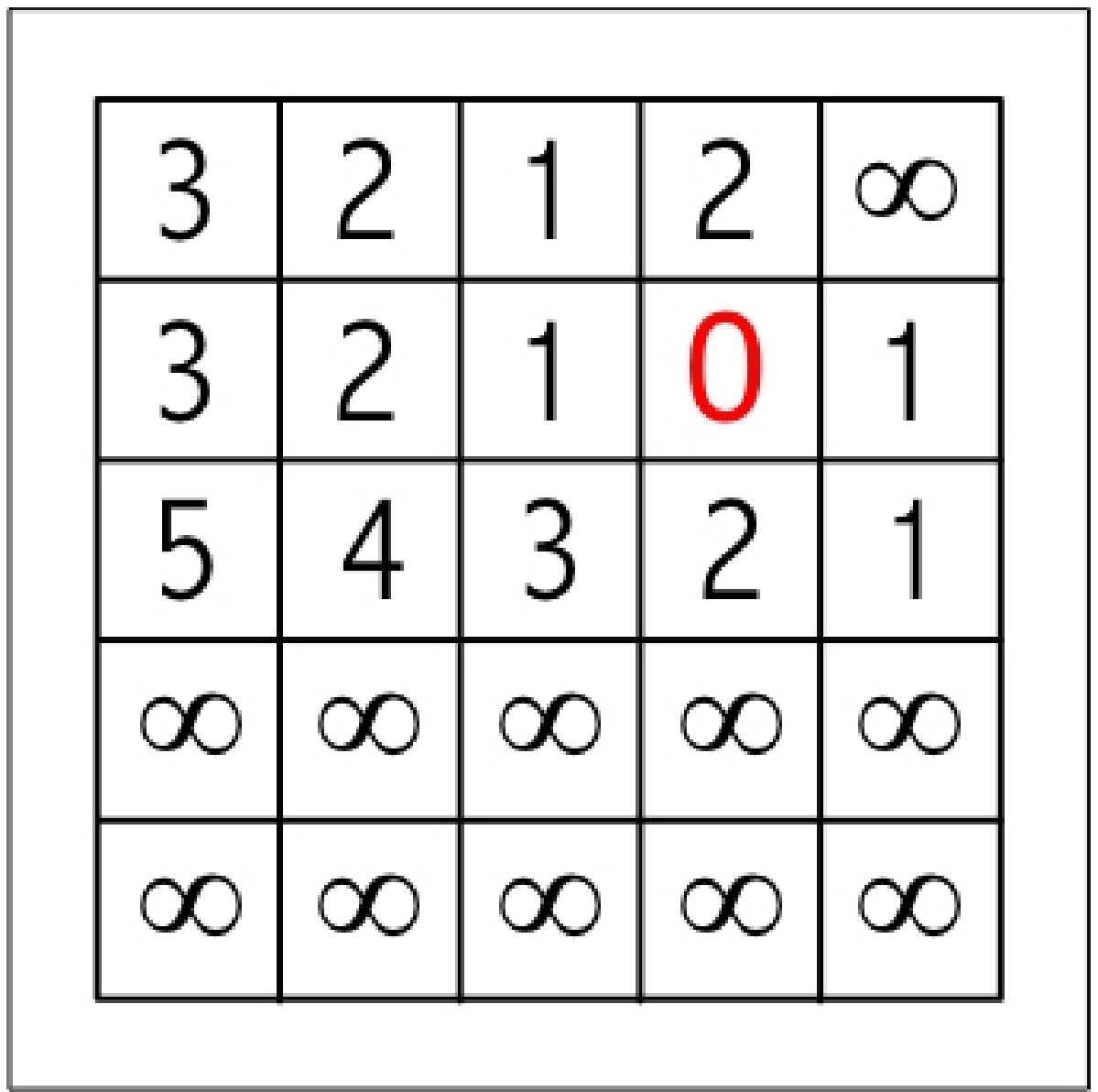}}
\subfigure[]{\includegraphics[height=0.28\linewidth]{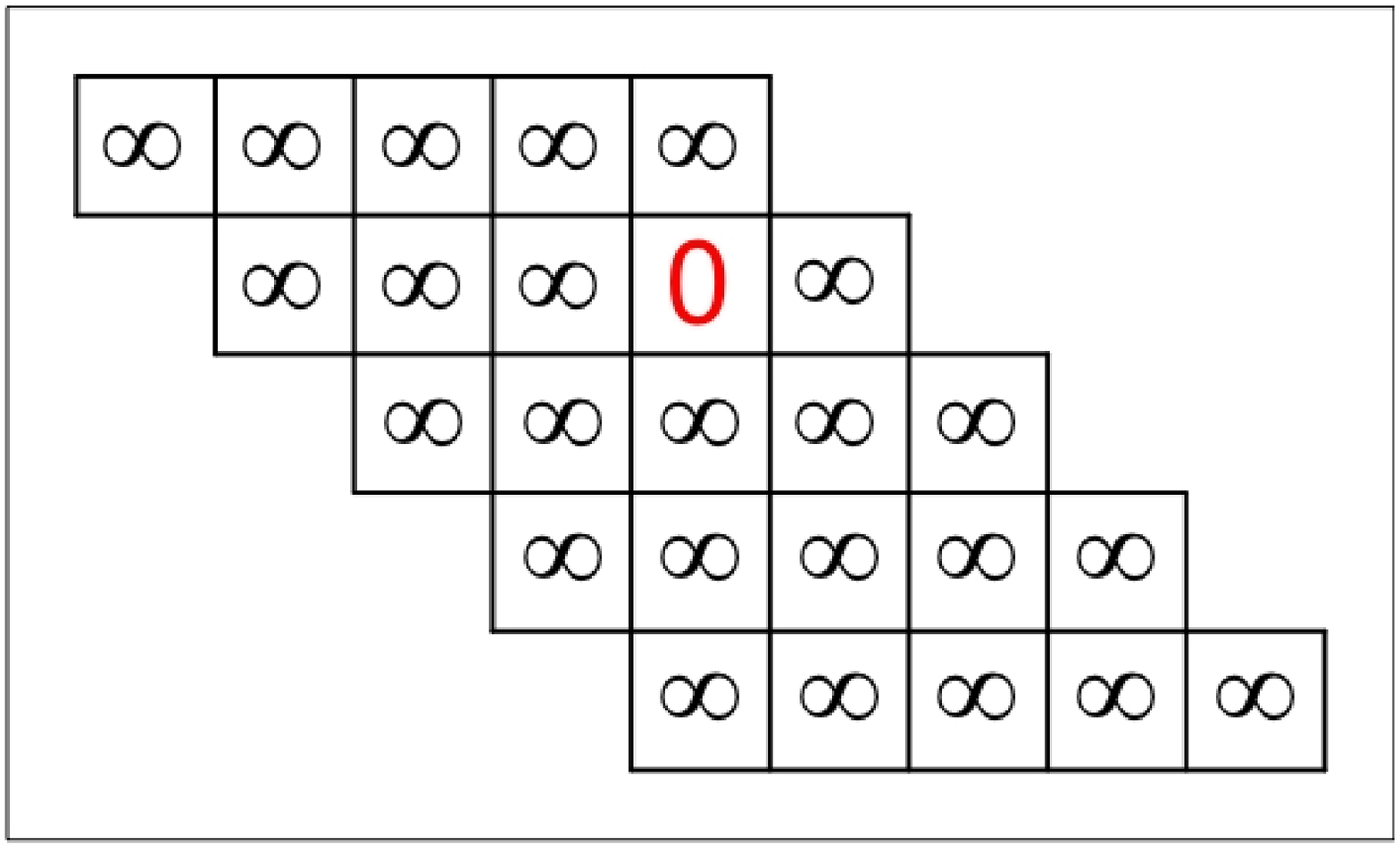}}
\subfigure[]{\includegraphics[height=0.28\linewidth]{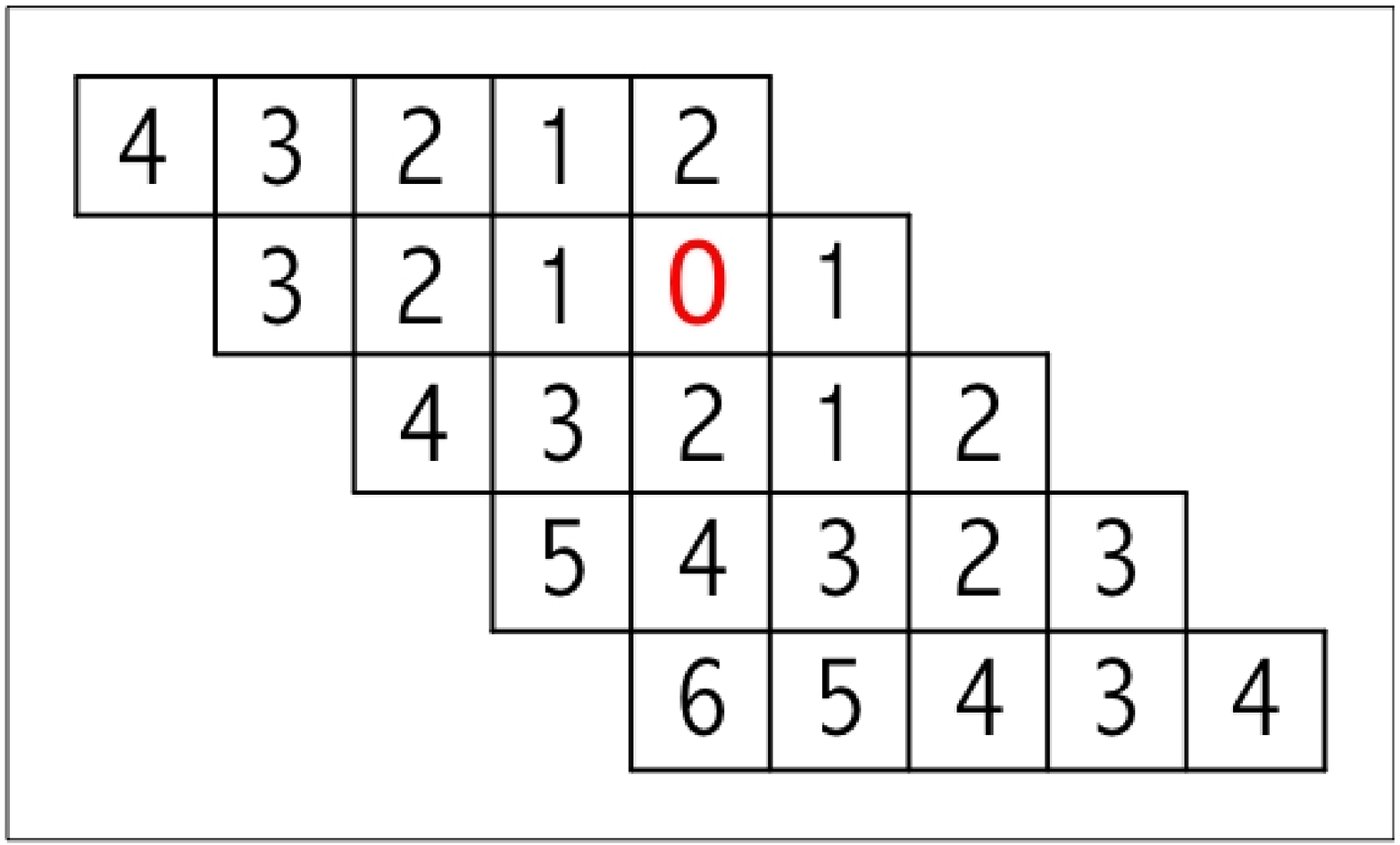}}
\subfigure[]{\includegraphics[height=0.28\linewidth]{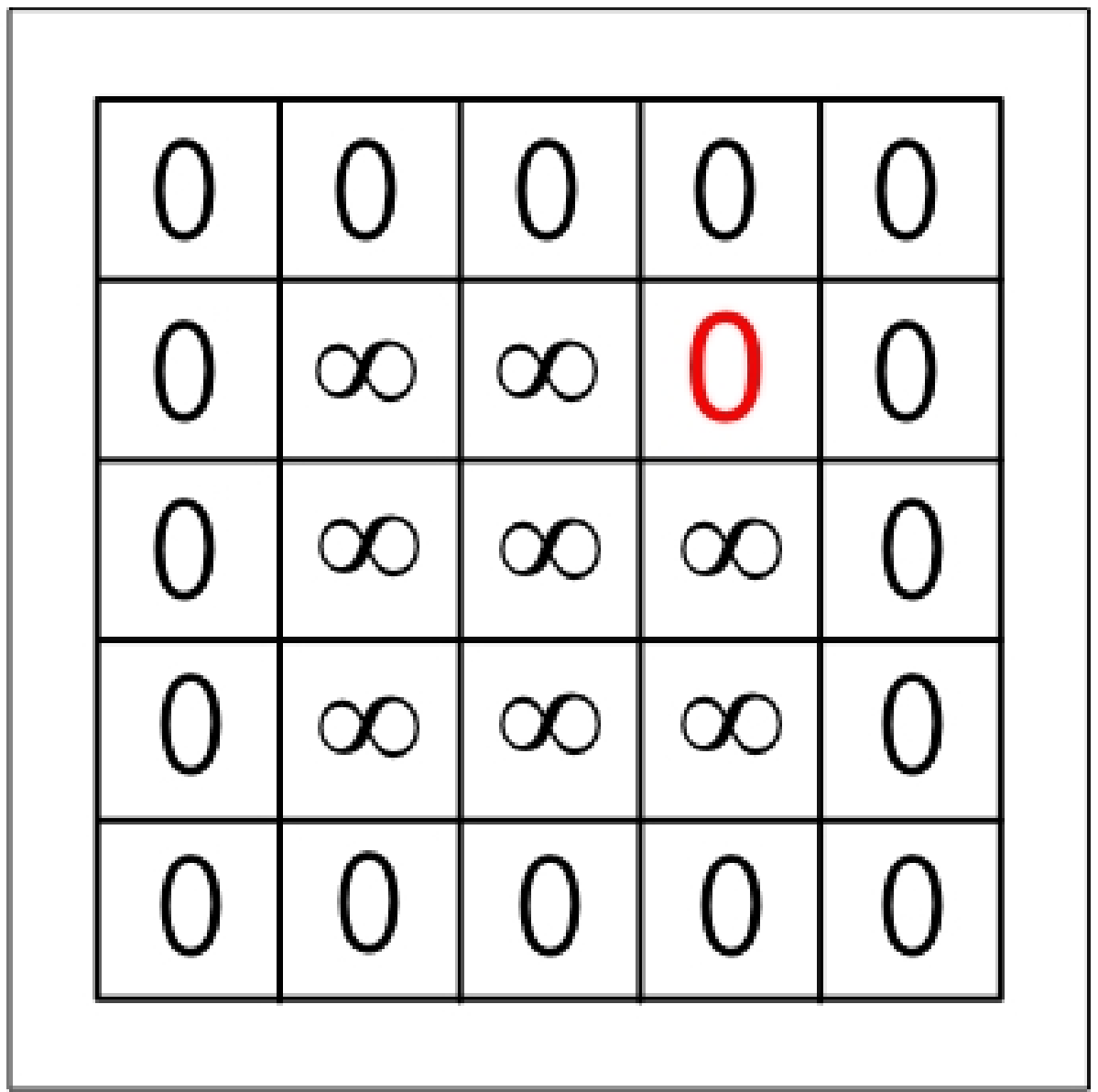}}
\subfigure[]{\includegraphics[height=0.28\linewidth]{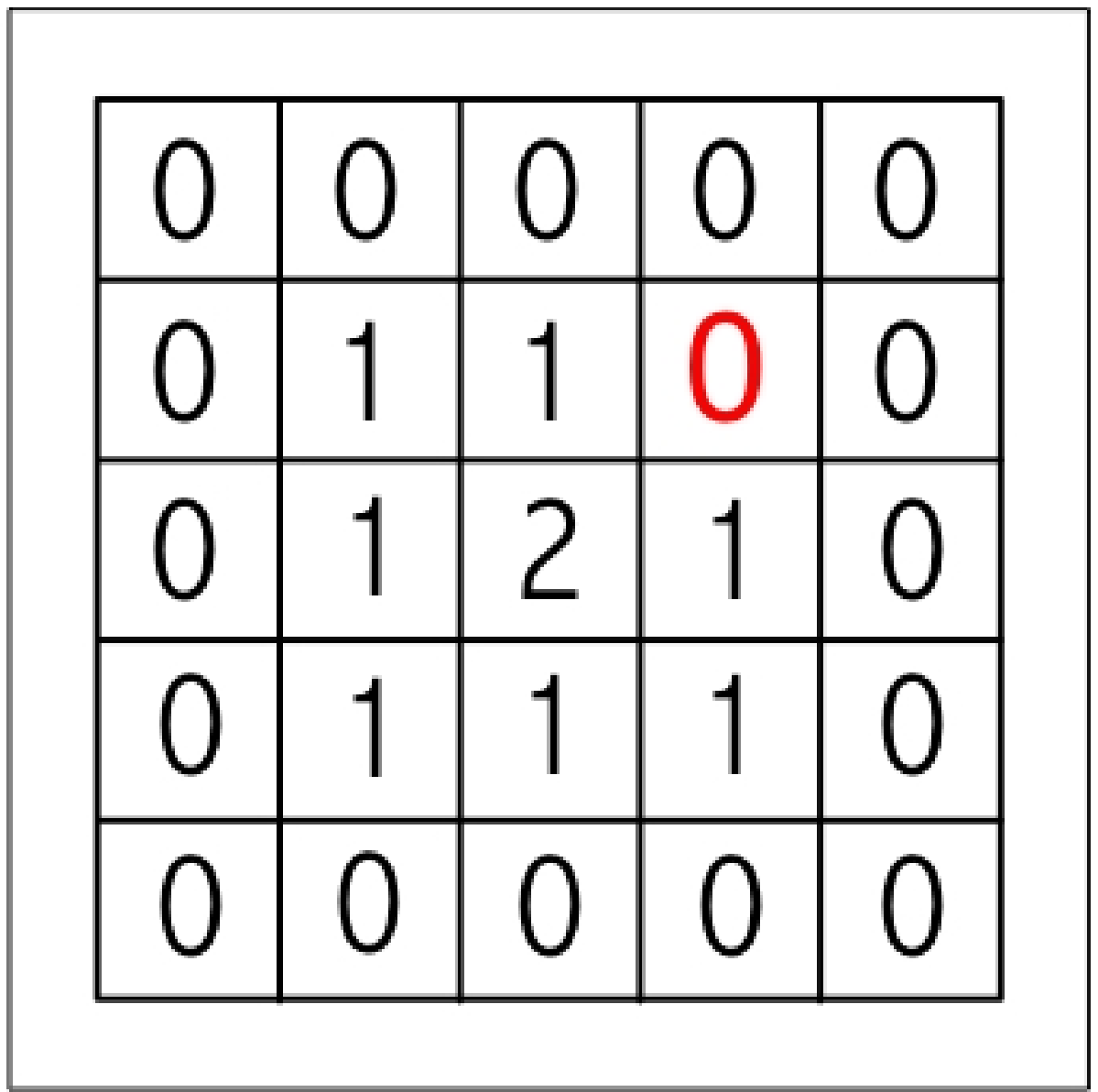}}
\caption{Consider the mask
$\mathcal{C}=\{((-1,0),1),((-1,1),1),((1,0),1),((1,-1),1)\}$ (a).
The image in (b) is not wedge-preserving when using this mask.  In
(c), the result of the chamfer algorithm using
$\mathcal{C}_1=\{((-1,0),1),((-1,1),1)\}$ and
$\mathcal{C}_2=\{((1,0),1),((1,-1),1)\}$ is shown. By considering the 
wedge-preserving image in (d), the correct distance map (shown in (e))
is produced. By assigning all border points to the background, the
image in (f) is achieved and the chamfer algorithm produces the
correct result shown in (g).}
\end{center}
\end{figure}

\begin{theorem}\label{chamf-alg-theorem}
If either
\begin{itemize}
\item for all border points $\point{p} \in \mathcal{S}, \point{p} \in \overline{X}$ or
\item $\mathcal{S}$ is wedge-preserving,
\end{itemize} then the chamfer algorithm in
Definition~\ref{chamfer_alg} produces distance maps as defined in
Definition~\ref{def:dist-map}. 
\end{theorem}

\begin{proof}
Let $\point{p}\in X$ and $\point{q}\in \overline{X}$ be such that
$d(\point{p},\point{q})=d(\point{p}, \overline{X})$ and
$\mathcal{I}=\left\{ i:\left(\vect{v}_i,w_i\right)\in
\mathcal{C}\right\}$. There is a set $\{\alpha_i\in\mathcal{R}\}$ such
that the point $\point{p}$ can be written $\displaystyle
\point{p}=\point{q}+\sum_{i\in \mathcal{I}} \alpha_i \vect{v}_i$ and
$\displaystyle d(\point{p}, \point{q})=\sum_{i\in \mathcal{I}}
\alpha_i w_i$. 

Let $f$ be the image after the first scan. By
Lemma~\ref{chamfer_lemma}, using that the masks support the scanning
order, the local distances are propagated in the following way. Also,
by Lemma~\ref{border_lemma} and Lemma~\ref{lem_wedge_equival}, all
points below are in $\mathcal{S}$. 
\begin{eqnarray}
f(\point{q}+\vect{v}_{k_1})=& w_{k_1} & \textrm{ for any $k_1$ such that}\nonumber\\
                            &         &\vect{v}_{k_1} \in
                            \mathcal{C}_1 \textrm{, and }
                            \alpha_{k_1}\geq 1\nonumber\\ 
f(\point{q}+\vect{v}_{k_1}+\vect{v}_{k_2})=& w_{k_1} + w_{k_2}
&\textrm{ for any $k_1$, $k_2$ such that } \nonumber\\
 & & \vect{v}_{k_1},\vect{v}_{k_2} \in \mathcal{C}_1 \nonumber\\ 
 & & \textrm{ and } \alpha_{k_1},\alpha_{k_2} \geq 1 \nonumber\\
 & & (\alpha_{k_1} \geq 2 \textrm{ if } k_1 = k_2),  \nonumber\\ 
&\vdots\nonumber\\
\displaystyle
f\left(\point{q}+\sum_{k_i:\vect{v}_{k_i}\in\mathcal{C}_1}
\alpha_{k_i} \vect{v}_{k_i} \right) = &\displaystyle
\sum_{k_i:w_{k_i}\in\mathcal{C}_1} \alpha_{k_i} w_{k_i}.\nonumber 
\end{eqnarray}
Let $g$ be the image after the second scan. Using the notation
$\displaystyle
\point{q}_N=\point{q}+\sum_{k_i:\vect{v}_{k_i}\in\mathcal{C}_1}
\alpha_{k_i} \vect{v}_{k_i}$, we get 
\begin{eqnarray}
g(\point{q}_N+\vect{v}_{l_1})=& g(\point{q}_N)+w_{l_1} &\textrm{ for
any $l_1$ such that }\nonumber\\
 & & \vect{v}_{l_1} \in \mathcal{C}_2 \textrm{, and }
\alpha_{l_1}\geq 1\nonumber\\ 
g(\point{q}_N+\vect{v}_{l_1}+\vect{v}_{l_2})=& g(\point{q}_N)+w_{l_1}
+ w_{l_2} &\textrm{ for any $l_1$, $l_2$ such}\nonumber\\
 & & \textrm{that }\vect{v}_{l_1},\vect{v}_{l_2} \in \mathcal{C}_2\nonumber\\
 & & \textrm{and } \alpha_{l_1},\alpha_{l_2} \geq 1\nonumber\\
 & & (\alpha_{l_1}\geq 2 \textrm{ if } l_1 = l_2)\nonumber\\ 
&\vdots\nonumber\\ 
\displaystyle
g\left(\point{q}_N+\sum_{l_i:\vect{v}_{l_i}\in\mathcal{C}_1}
\alpha_{l_i} \vect{v}_{l_i} \right) = &g(\point{q}_N)+\displaystyle
\sum_{l_i:w_{l_i}\in\mathcal{C}_2} \alpha_{l_i} w_{l_i}= & \sum
\alpha_i w_i.\nonumber 
\end{eqnarray}
This proves that all steps of a shortest path between $\point{p}$ and
$\point{q}$ have been processed. Indeed, after the second scan, the
value of the point $\point{p}$ is $g(\point{p})=d(\point{p},
\point{q})=\sum \alpha_i w_i$. 
\end{proof}

\begin{remark}
In Figure~\ref{chamf-alg-conditions}, the need for the conditions in
Theorem~\ref{chamf-alg-theorem} are shown in example images. 
\end{remark}

\section{Best weights for the BCC and FCC grids}\label{sec:bccfcc}

The BCC grid ($\mathbb{B}$) and the FCC grid ($\mathbb{F}$) are
defined as follows: 
\begin{displaymath}
\mathbb{B} = \{ \point{p} (x, y, z) \in \mathbb{Z}^{3} \textrm{ and }
x \equiv y \equiv z \pmod{2} \} \textrm{ and}
\end{displaymath}
\begin{displaymath}
\mathbb{F} =\{\point{p}(x,y,z) \in \mathbb{Z}^3 \textrm{ and } x+y+z
\equiv 0 \pmod {2}\}. 
\end{displaymath}
A voxel is defined as the Voronoi region of a point $\point{p}$ in a
grid. Figure~\ref{neighborhoods}~(a) shows a voxel of a BCC grid. A
BCC voxel has two kinds of face-neighbors (but no edge- or
vertex-neighbors), which results in the 8-neighborhood
(Figure~\ref{neighborhoods}~(b)) and the 14-neighborhood
(Figure~\ref{neighborhoods}~(c)). On the FCC grid, each voxel (see
Figure~\ref{neighborhoods}~(d)) has 12 face-neighbors and 6
vertex-neighbors. The resulting 12- and 18-neighborhoods are shown in
(Figure~\ref{neighborhoods}~(e)) and (Figure~\ref{neighborhoods}~(f)),
respectively. 

\begin{figure}[ht!]
\begin{center}
\subfigure[]{\includegraphics[height=0.25\linewidth]{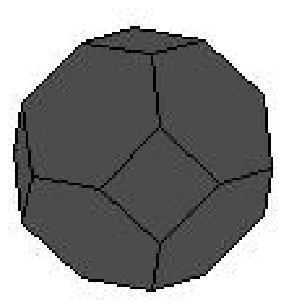}}
\subfigure[]{\includegraphics[height=0.25\linewidth]{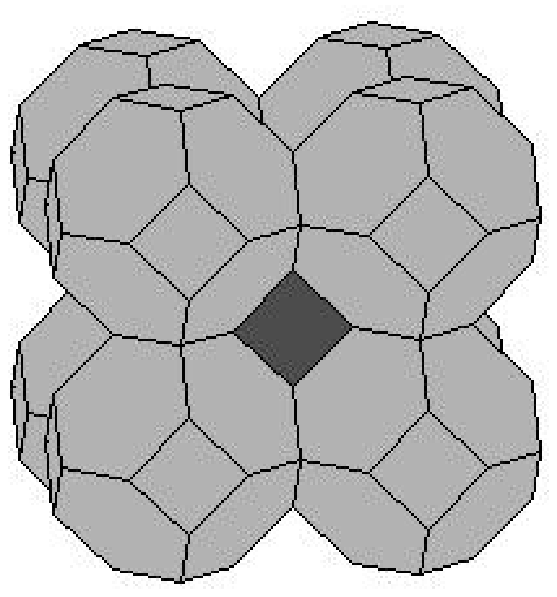}}
\subfigure[]{\includegraphics[height=0.25\linewidth]{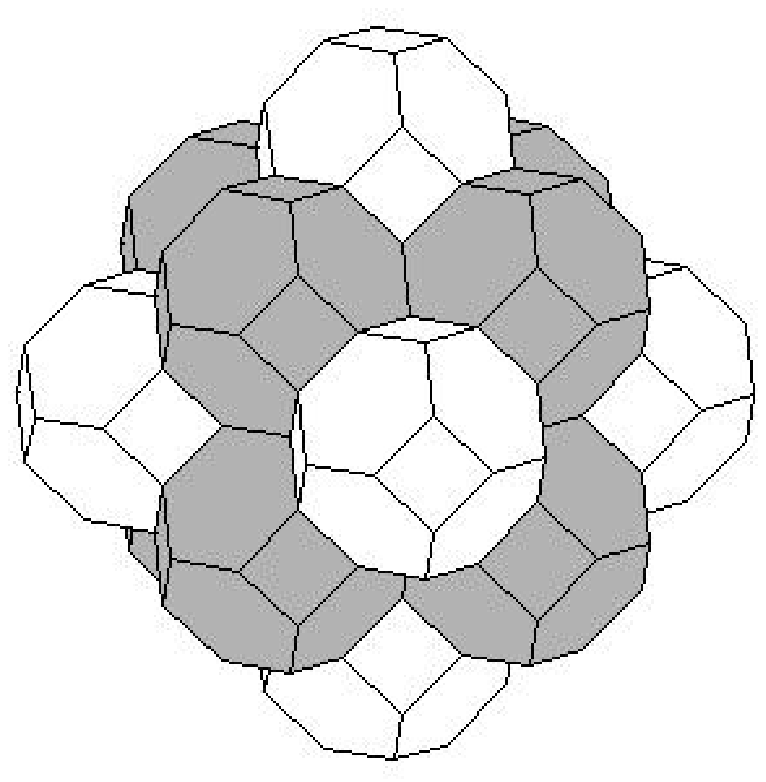}}
\subfigure[]{\includegraphics[height=0.25\linewidth]{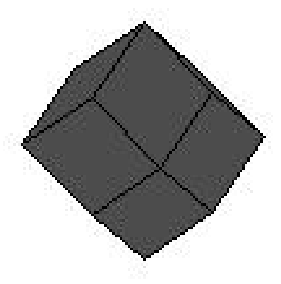}}
\subfigure[]{\includegraphics[height=0.25\linewidth]{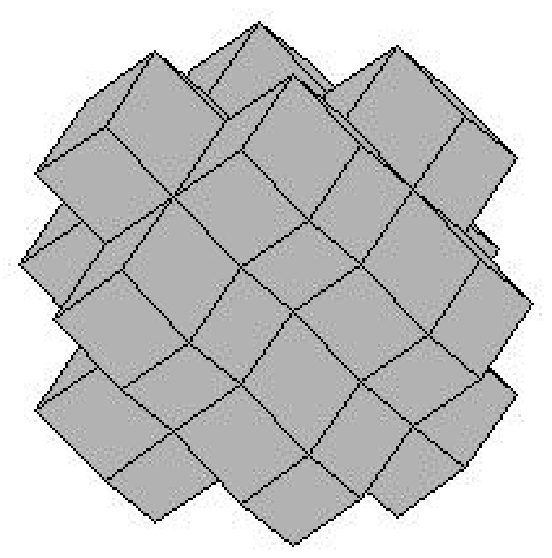}}
\subfigure[]{\includegraphics[height=0.25\linewidth]{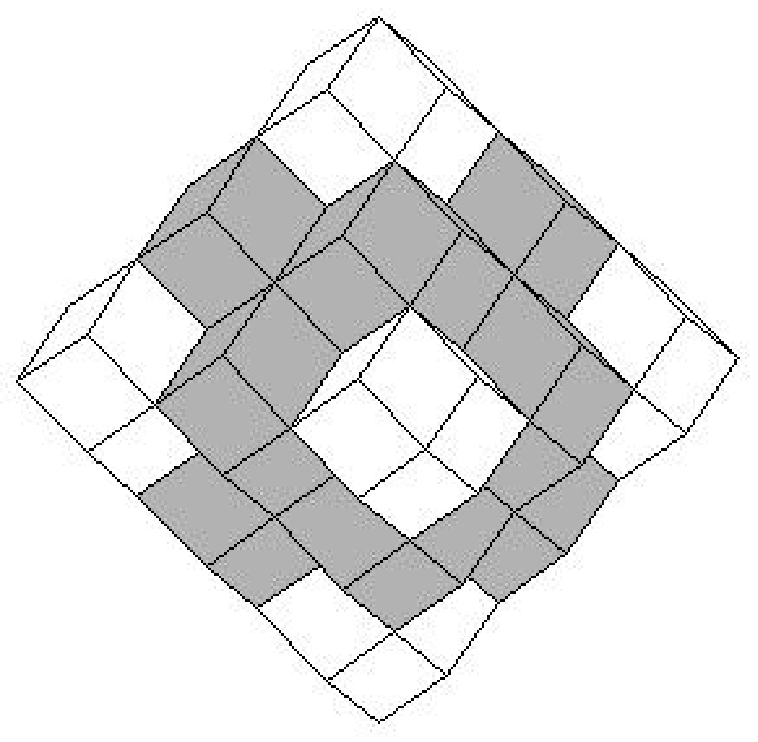}}
\caption{Different neighborhoods for the BCC grid ((a)-(c)) and the
FCC grid ((d)-(f)).\label{neighborhoods}} 
\end{center}
\end{figure}

Note that the high number of face neighbors 
(12 for the FCC grid and 14 for the BCC grid)
implies these grids are more compact than the cubic grid, which has
only $6$ face neighbors. 

\subsection{Results of the previous sections applied to BCC
grids}\label{sec:bcc_application}  
To apply the results of
Sections~\ref{sec:wdp}~and~\ref{sec:chamferAlgo} to the BCC grid, we
should check that: 
\begin{itemize}
  \item $\mathbb{B}$ is a sub-module of $\mathbb{R}^{3}$
  \item Masks decomposed in $\mathbb{B}$-basis wedges can be created
  \item A scanning order can be defined on the BCC lattice.
\end{itemize}

\begin{lemma}\label{lemma:bccsubmodule}
Let $ \mathbb{B}^{n} = \{ \point{p} : p^{1}, p^{2}, ..., p^{n} \in
\mathbb{Z} \textrm{ and } p^{1} \equiv p^{2} \equiv ... \equiv p^{n}
\pmod{2} \}$ , then 
  $\left( \mathbb{B}^{n}, \mathbb{Z}, +, \times \right)$ is a
  sub-module of $\mathbb{R}^{n}$.
\end{lemma}
\begin{proof}
  Let us prove first that $\mathbb{B}^{n}$ is an Abelian group.
  Given $\point{p}$, $\point{q} \in \mathbb{B}^{n}$, there exist 
  $\alpha_{2}, ..., \alpha_{n}$, and $\beta_{2}, ..., \beta_{n} \in
  \mathbb{Z}$ 
  such that $\forall k \in [2..n]\ p^{k} = p^{1} + 2 \alpha_{k}$ and
  $q^{k} = q^{1} + 2 \beta_{k}$. The point $ \point{r} = \point{p} +
  \point{q} = \point{q} + \point{p}$ also belongs to
  $\mathbb{B}^{n}$. Indeed, 
  $\forall k \in [2..n]$, $r^{k} = p^{k} + q^{k} =
  p^{1} + 2 \alpha_{k} + q^{1} + 2 \beta_{k} = r^{1} + 2 \gamma_{k}$
  with $\gamma_{k} = \alpha_{k} + \beta_{k} \in \mathbb{Z}$.   
  Moreover, $-\point{p} \in \mathbb{B}^{n}$ as $-p^{k} = -p^{1} + 2
  (-\alpha_{k})$ with $-\alpha_{k} \in \mathbb{Z}$. 
  We now have to prove that $\forall \lambda \in \mathbb{Z}$, $\forall
  \point{p} \in \mathbb{B}^{n}$, $\lambda \cdot \point{p} \in
  \mathbb{B}^{n}$. Indeed, $\forall k \in [1..n], \lambda p^{k} \in
  \mathbb{Z}$ and $\forall k \in [2..n] \lambda p^{k} = \lambda p^{1}
 + 2 (\lambda \alpha_{k}) \textrm{ with } (\lambda \alpha_{k}) \in
   \mathbb{Z}$. 
\end{proof}
This result is true for all $n \in \mathbb{N}\setminus \{0\}$ and
obviously for $n = 3$.

\subsubsection{Chamfer mask geometry}\label{sec:bcc_mask_geometry}
\begin{lemma}
\label{lemma:bccbasis}
  A family $\mathcal{F} = \left( \vect{v}_{1}, \vect{v}_{2},
    \vect{v}_{3} \right)$ is a basis of $\mathbb{B}$ iff
  \begin{displaymath}
    \Delta_{\mathcal{F}}^{0} = \det(\vect{v}_{1}, \vect{v}_{2},
    \vect{v}_{3}) = \pm 4.
  \end{displaymath}
\end{lemma}
\begin{proof}
First, we prove that any determinant of 3 vectors $\vect{v}_{1},
\vect{v}_{2}, \vect{v}_{3}$ of $\mathbb{B}$ is
a multiple of $4$. Indeed, given a vector $\vect{v}_{i} (x_{i}, y_{i},
z_{i})$, $\vect{v}_{i} \in \mathbb{B}$ iff there exists $\alpha_{i},
\beta_{i} \in \mathbb{Z}$ such that $y_{i} = x_i + 2 \alpha_{i}$ and $z_{i}
= x_i + 2 \beta_{i}$. Then the determinant of $\vect{v}_{1}, \vect{v}_{2}$
and $\vect{v}_{3}$ is
\begin{displaymath}
\begin{array}{l c c c}
  \det(\vect{v}_{1},\vect{v}_{2}, \vect{v}_{3}) & = & 
\left| \begin{array}{c c c} 
x_{1} & x_{2} & x_{3} \\
y_{1} & y_{2} & y_{3} \\
z_{1} & z_{2} & z_{3} \\ \end{array} \right| 
& =  \left| \begin{array}{c c c} 
x_{1} & x_{2} & x_{3} \\
x_{1} + 2 \alpha_{1} & x_{2} + 2 \alpha_{2} & x_{3} + 2 \alpha_{3} \\
x_{1} + 2  \beta_{1} & x_{2} + 2  \beta_{2} & x_{3} + 2  \beta_{3} \\ 
\end{array} \right| \\
 & = 4 & 
\underbrace{
 \left| \begin{array}{c c c} 
x_{1} & x_{2} & x_{3} \\
\alpha_{1} & \alpha_{2} & \alpha_{3} \\
 \beta_{1} &  \beta_{2} &  \beta_{3} \\ \end{array} \right|} & \\

 & & \in \mathbb{Z} & \\ 
\end{array}. 
\end{displaymath}

With the notations of paragraph \ref{sec:wdp} we obtain for any vector
$\vect{v}(x, y, z) = (x, x + 2\alpha, x + 2\beta)$ of $\mathbb{B}$: 
\begin{displaymath}
\frac{1}{\Delta_{\mathcal{F}}^{0}} \times
\Delta_{\mathcal{F}}^{1}(\vect{x}) = 
\frac{1}{4 \left| \begin{array}{c c c} 
x_{1} & x_{2} & x_{3} \\
\alpha_{1} & \alpha_{2} & \alpha_{3} \\
 \beta_{1} &  \beta_{2} &  \beta_{3} \\ \end{array} \right|} 
\times  
4 \left| \begin{array}{c c c} 
x & x_{2} & x_{3} \\
\alpha & \alpha_{2} & \alpha_{3} \\
 \beta &  \beta_{2} &  \beta_{3} \\ \end{array} \right|
\end{displaymath}

and $\frac{1}{\Delta_{\mathcal{F}}^{0}} \times
\Delta_{\mathcal{F}}^{1}(\vect{x}) \in \mathbb{Z}$ iff 
$\left| \begin{array}{c c c} 
x_{1} & x_{2} & x_{3} \\
\alpha_{1} & \alpha_{2} & \alpha_{3} \\
 \beta_{1} &  \beta_{2} &  \beta_{3} \\ \end{array} \right| = \pm 1$
 which means $\Delta_{\mathcal{F}}^{0} = \pm 4$.
The same result applies for $\frac{1}{\Delta_{\mathcal{F}}^{0}} \times
\Delta_{\mathcal{F}}^{2}(\vect{x})$ and
$\frac{1}{\Delta_{\mathcal{F}}^{0}} \times
\Delta_{\mathcal{F}}^{3}(\vect{x})$. 
By Lemma \ref{lemma:value} we obtain $\mathcal{F}$ is a basis of
$\mathbb{B}$ iff $\Delta_{\mathcal{F}}^{0} = \pm 4$.
\end{proof}

To compute a weighted distance map on an image stored on a BCC grid,
one can consider a mask built using 8-neighbors: 
\begin{displaymath}
\mathcal{C}_{8} = \left\{(\vect{v}_{i} (\pm 1, \pm 1, \pm 1), w_{1}),  \right\}
\end{displaymath}
This mask contains only one type of weight as the Euclidean
distance between each face-sharing neighbor to the central voxel is
the same.

This mask can be decomposed into $\mathbb{B}$-basis sectors. Indeed,
let us consider the wedge $\mathcal{F}_{\mathbb{B}} =
\llangle \vect{v}_{1} (1, 1, 1), \vect{v}_{2} (-1, 1, 1), \vect{v}_{3}
(1, -1, 1) \rrangle$, and its symmetric wedges.
\begin{displaymath}
\Delta_{\mathcal{F}_{\mathbb{B}}}^{0} = 
\det (\vect{v}_{1}, \vect{v}_{2}, \vect{v}_{3}) = 
\left| \begin{array}{c c c}  
1 & -1 &  1 \\ 
1 &  1 & -1 \\ 
1 &  1 &  1 \\ 
\end{array} \right|
 = 4 
\end{displaymath}
and from Lemma \ref{lemma:bccbasis}, $\mathcal{F}_{\mathbb{B}}$ is a
$\mathbb{B}$-basis sector.

If we want larger masks, we can split the
wedge $\mathcal{F}_{\mathbb{B}}$ and its symmetric wedges according to the Farey triangulation technique \cite{fouard:ivc:2005,remy:iwcia:2000}. 
To do so, given a
$\mathbb{B}$-basis wedge $\mathcal{F} = \llangle \vect{v}_{1},
\vect{v}_{2}, \vect{v}_{3} \rrangle$, choose an edge to be split,
say the edge between $\vect{v}_{1}$ and $\vect{v}_{2}$. 
We create $\vect{v}(x, y, z)$ such that $\vect{v} = \vect{v}_{1}
\oplus  \vect{v}_{2}$ defined by $x = x_{1} + x_{2}$, $y = y_{1} +
y_{2}$ and $z = z_{1} + z_{2}$. The two obtained wedges
$\mathcal{F}_{1} = \llangle \vect{v},\vect{v}_{2}, \vect{v}_{3}
\rrangle$ and $\mathcal{F}_{2} = \llangle \vect{v}_{1},\vect{v},
\vect{v}_{3} \rrangle$ are $\mathbb{B}$-basis sectors. Indeed, 
\begin{displaymath}
\begin{array}{l c l}
\Delta_{\mathcal{F}_{1}}^{0} & = & 
\det (\vect{v}, \vect{v}_{2}, \vect{v}_{3}) = 
\left| \begin{array}{c c c}  
x_{1} + x_{2} & x_{2} & x_{3} \\ 
y_{1} + y_{2} & y_{2} & y_{3} \\ 
z_{1} + z_{2} & z_{2} & z_{3} \\ 
\end{array} \right| = 
\left| \begin{array}{c c c}  
x_{1} & x_{2} & x_{3} \\ 
y_{1} & y_{2} & y_{3} \\ 
z_{1} & z_{2} & z_{3} \\ 
\end{array} \right| + 
\left| \begin{array}{c c c}  
x_{2} & x_{2} & x_{3} \\ 
y_{2} & y_{2} & y_{3} \\ 
z_{2} & z_{2} & z_{3} \\ 
\end{array} \right| \\
 & = & \Delta_{\mathcal{F}}^{0} = \pm 4.\\
\end{array}
\end{displaymath}
In the same way, $\Delta_{\mathcal{F}_{2}}^{0} = \pm 4$. 

If we apply the first steps of this method, we can retrieve the
chamfer mask using the 14-neighborhood proposed by
\cite{strand:cviu:2005}.
But we can also build larger masks by splitting the new wedges.

\subsubsection{Chamfer mask weights}
To compute optimal integer chamfer weights for BCC grid the  depth-first
search method of \cite{fouard:ivc:2005,malandain:RR:2005}\footnote{The
corresponding code is available at
www.cb.uu.se/$\sim$tc18/code-data-set} can be used. This method
computes weights sets leading to weighted distance maps with small
relative error with respect to the corresponding Euclidean
map. Moreover, it keeps a weights set only if the normalized polytope
of the corresponding mask is convex, which allows application of the
results from Section~\ref{sec:wdp}. 

The following tables give sets of integer chamfer mask weights for
BCC grids with the scale factor allowing comparison of the final weighted
distance map with an Euclidean one composed with real numbers. They
also give the maximum relative error that can occur between the
weighted distance map and the Euclidean one.

For the two first cases, we also give the optimal real weights
computed using the equations of Section~\ref{sec:coeffsopt}.

Balls obtained by using different number of weights (shown in bold
below) are shown in Figure~\ref{balls}. 

\begin{enumerate}
\item One weight\\
For the mask corresponding to the 8-neighborhood, using only one
weight, the equations of section \ref{sec:coeffsopt} give the
following result:\\
\begin{tabular}{| l l |}
\hline
real optimal weight    & 1.268 \\
real optimal error (\%)     & 26.79 \\
\hline
integer optimal weight & 1     \\
real map scale factor       & 1.268 \\ 
\hline
\end{tabular}\\
These results are consistent with \cite{strand:cviu:2005}.\\

\item Two weights\\
\begin{scriptsize}
\begin{tabular}{| c c | c | c c c c c c c c | }
\hline
vector  & weight  & real  & \multicolumn{8}{c |}{weights} \\
\hline
(1 1 1) & $w_{1}$ & 1.547 & 1 & 2 & 3 & 4 & 5 & 6 & \bf{13} & 19 \\
(2 0 0) & $w_{2}$ & 1.786 & 2 & 3 & 4 & 5 & 6 & 7 & \bf{15} & 22 \\
\hline
\multicolumn{2}{| c |}{Scale factor} & 1 & 1.268 & 0.731 & 0.504 & 0.383 &
0.308 & 0.256 & \bf{0.119} & 0.081 \\
\hline
\multicolumn{2}{| c |}{Error (\%)} & 10.69 & 26.79 & 15.59 & 12.70 & 11.60
& 11.07 & 10.78 & \bf{10.72} & 10.71 \\
\hline
\end{tabular}\\
\end{scriptsize}
\ \\

\item Three weights\\
\begin{scriptsize}
\begin{tabular}{| c c | c c c c c c c c c | }
\hline 
vector  & weight  & \multicolumn{9}{c |}{weights} \\ 
\hline
(1 1 1) & $w_{1}$ & 1 & 2 & 4 & 5 & 6  & \bf{13} & 19 & 26 & 33\\
(2 0 0) & $w_{2}$ & 2 & 2 & 5 & 6 & 7  & \bf{15} & 22 & 30 & 38\\
(2 2 0) & $w_{3}$ & 2 & 3 & 7 & 8 & 10 & \bf{22} & 31 & 43 & 54\\
\hline
\multicolumn{2}{| c |}{Scale factor} & 
1.268 & 0.899 & 0.396 & 0.325 & 0.270 & \bf{0.125} & 0.0857 & 0.0626 &
      0.0494\\
\hline
\multicolumn{2}{| c |}{Error (\%)} & 
26.79 & 10.10 & 8.50  & 7.94  & 6.39  & \bf{6.34} & 6.12   & 6.12   &
      6.11\\
\hline
\multicolumn{11}{| c |}{Optimal error (\%): 6.02}\\
\hline
\end{tabular}\\
\end{scriptsize}

Note: for the previous masks, we displayed real optimal weights
to be able to compare our results with previous papers. 
However, as real optimal weights are computed sector by sector,
they may change from one sector to another for the same mask
vector. We do not display them for several sectors as they may not be
consistent. \\

\item Four weights\\
\begin{scriptsize}
\begin{tabular}{| c c | c c c c c c c c | }
  \hline
  vector  & weight  & \multicolumn{8}{c |}{weights} \\
  \hline
  (1 1 1) & $w_{1}$ & 1 & 2 & 4 & 5  & 6  & 9  & \bf{15} & 26\\
  (2 0 0) & $w_{2}$ & 2 & 2 & 4  & 6  & 7  & 10  & \bf{17} & 29\\
  (2 2 0) & $w_{3}$ & 2 & 3 & 6  & 8  & 10  & 14 & \bf{24} & 41\\
  (3 1 1) & $w_{4}$ & 3 & 4 & 7  & 10  & 12 & 17 & \bf{29} & 50\\
  \hline
  \multicolumn{2}{| c |}{Scale factor} & 
  1.268 & 0.899 & 0.460 & 0.334 & 0.275 & 0.194 & \bf{0.113} & 0.0662\\
  \hline
  \multicolumn{2}{| c |}{Error (\%)} & 
  26.79 & 10.10 & 7.94  & 5.57  & 4.73  & 4.21  & \bf{4.00} & 3.99\\
  \hline
  \multicolumn{10}{| c |}{Optimal error (\%): 3.96}\\
  \hline
\end{tabular}\\
\end{scriptsize}

\end{enumerate}

\subsection{Results of previous sections applied to FCC grids}
To apply the results of
Sections~\ref{sec:wdp}~and~\ref{sec:chamferAlgo} to the FCC grid, we
should again check that: 
\begin{itemize}
  \item $\mathbb{F}$ is a sub-module of $\mathbb{R}^{3}$
  \item Masks decomposed in $\mathbb{F}$-basis wedges can be created
  \item A scanning order can be defined on the FCC lattice.
\end{itemize}
The calculations are similar to those in Section~\ref{sec:bcc_application}.
\begin{lemma}\label{lemma:fccsubmodule}
Let $ \mathbb{F}^{n} = \{ \point{p} : p^{1}, p^{2}, ..., p^{n} \in
\mathbb{Z} \textrm{ and } p^{1}+p^{2}+p^{3} \equiv 0 \pmod {2} \}$ , then 
  $\left( \mathbb{F}^{n}, \mathbb{Z}, +, \cdot \right)$ is a
  sub-module of $\mathbb{R}^{n}$.
\end{lemma}

\begin{lemma}
\label{lemma:fccbasis}
  A family $\mathcal{F} = \left( \vect{v}_{1}, \vect{v}_{2},
    \vect{v}_{3} \right)$ is a basis of $\mathbb{F}$ iff
  \begin{displaymath}
    \Delta_{\mathcal{F}}^{0} = \det(\vect{v}_{1}, \vect{v}_{2},
    \vect{v}_{3}) = \pm 2.
  \end{displaymath}
\end{lemma}

These Lemmas can be proved similar to
Lemma~\ref{lemma:bccsubmodule}~and~\ref{lemma:bccbasis} by noting that
for any $\point{p}\in\mathbb{F}^n$, there exist $\alpha \in\mathbb{Z}$
such that $p^{n} = p^{1} + \dots + p^{n-1} +2 \alpha$. 

To compute the weighted distance on FCC, we first consider the
smallest mask, i.e., the mask containing 12-neighbors: 
\begin{displaymath}
\mathcal{C}_{12} = \left\{(\vect{v}_{i} (\pm 1, \pm 1, 0),
w_{1}),(\vect{v}_{i} (\pm 1, 0 ,\pm 1), w_{1}),(\vect{v}_{i} (0, \pm
1, \pm 1), w_{1}),  \right\} 
\end{displaymath}

The wedges 
$\mathcal{F}_{\mathbb{F}}^1 =
\llangle \vect{v}_{1} (1, 1, 0), \vect{v}_{2} (1, 0, 1), \vect{v}_{3}
(1, -1, 0) \rrangle$,\\ 
$\mathcal{F}_{\mathbb{F}}^2 =
\llangle \vect{v}_{1} (1, 1, 0), \vect{v}_{2} (1, 0, 1), \vect{v}_{4}
(2, 0, 0) \rrangle$, and their symmetric wedges are $\mathbb{F}$-basis
sectors. 
\begin{displaymath}
\Delta_{\mathcal{F}^1_{\mathbb{F}}}^{0} = 
\Delta_{\mathcal{F}^2_{\mathbb{F}}}^{0} = 2
\end{displaymath}
and from Lemma \ref{lemma:fccbasis}, $\mathcal{F}_{\mathbb{F}}^1$ and
$\mathcal{F}_{\mathbb{F}}^2$ are $\mathbb{F}$-basis sectors. 

The splitting of the sectors are analogous to
Section~\ref{sec:bcc_mask_geometry}, but the determinants equals $\pm
2$ instead of $\pm 4$. 

We get the following mask weights. Balls obtained by using different
number of weights (shown in bold below) are shown in
Figure~\ref{balls}.

\begin{enumerate}
\item One weight\\
For the mask corresponding to the 8-neighborhood, using only one
weight, the equations of section \ref{sec:coeffsopt} give the
following result:\\
\begin{tabular}{| l l |}
\hline
real optimal weight    & 1.172 \\
real optimal error (\%)     & 17.16 \\
\hline
integer optimal weight & 1     \\
real map scale factor       & 1.172 \\ 
\hline
\end{tabular}\\
These results are consistent with \cite{strand:cviu:2005}.\\

\item Two weights\\
\begin{scriptsize}
\begin{tabular}{| c c | c | c c c | }
\hline
vector  & weight  & real  & \multicolumn{3}{c |}{weights} \\
\hline
(1 1 0) & $w_{1}$ & 1.271 & 1 & 1 & \bf{2}\\
(2 0 0) & $w_{2}$ & 1.798 & 1 & 2 & \bf{3}\\
\hline
\multicolumn{2}{| c |}{Scale factor} &
1 & 1.464 & 1.172 &\bf{0.636}\\
\hline
\multicolumn{2}{| c |}{Error (\%)} &
10.10 & 26.79 & 17.16 & \bf{10.10}\\
\hline
\end{tabular}\\
\end{scriptsize}
\ \\

\item Three weights\\
\begin{scriptsize}
\begin{tabular}{| c c | c c c c c c c c c | }
\hline
vector  & weight  & \multicolumn{9}{c |}{weights} \\
\hline
(1 1 0) & $w_{1}$ & 1 & 1 & 2 & 2 & 4  & 6 & 7 & \bf{11} & 15\\
(2 0 0) & $w_{2}$ & 1 & 2 & 3 & 3 & 6 & 9 & 10 & \bf{16} & 22\\
(2 1 1) & $w_{3}$ & 2 & 2 & 3 & 4 & 7 & 10 & 12 & \bf{19} & 26\\
\hline
\multicolumn{2}{| c |}{Scale factor} & 
1.464 & 1.172 & 0.694 & 0.636 & 0.325 & 0.226 & 0.191 & \bf{0.121} & 0.0887\\
\hline
\multicolumn{2}{| c |}{Error (\%)} & 
26.79 & 17.16 & 15.04  & 10.10  & 7.94  & 7.76 & 6.19 & \bf{6.16} & 5.95\\
\hline
\multicolumn{11}{| c |}{Optimal error (\%): 5.93}\\
\hline
\end{tabular}\\
\end{scriptsize}
\ \\
\item Four weights\\
\begin{scriptsize}
\begin{tabular}{| c c | c c c c c c c c | }
  \hline
  vector  & weight  & \multicolumn{8}{c |}{weights} \\
  \hline
  (1 1 0) & $w_{1}$ & 1 & 1 & 2 & 3 & 5  & 5 & 9 & \bf{12}\\
  (2 0 0) & $w_{2}$ & 2 & 2 & 3 & 4 & 7  & 7 & 13 & \bf{17}\\
  (2 1 1) & $w_{3}$ & 2 & 2 & 4  & 5  & 9  & 9 & 16 & \bf{21}\\
  (2 2 2) & $w_{4}$ & 2 & 3 & 5  & 7  & 12 & 13 & 23 & \bf{30}\\
  \hline
  \multicolumn{2}{| c |}{Scale factor} & 
  1.268 & 1.172 & 0.651 & 0.472 & 0.274 & 0.272 & 0.150 & \bf{0.113}\\
  \hline
  \multicolumn{2}{| c |}{Error (\%)} & 
  26.79 & 17.16 & 7.94  & 5.57  & 5.15  & 4.64  & 4.63 & \bf{4.07}\\
  \hline
  \multicolumn{10}{| c |}{Optimal error (\%): 3.98}\\
  \hline
\end{tabular}\\
\end{scriptsize}

\end{enumerate}

\begin{figure}[ht!]
\begin{center}
\subfigure[]{\includegraphics[width=0.35\linewidth]{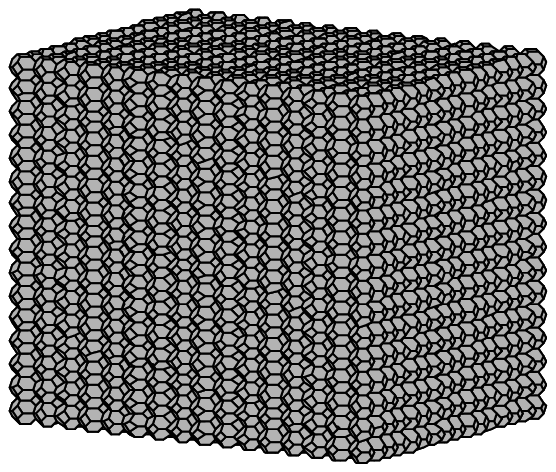}}
\hfill
\addtocounter{subfigure}{+3}
\subfigure[]{\includegraphics[width=0.35\linewidth]{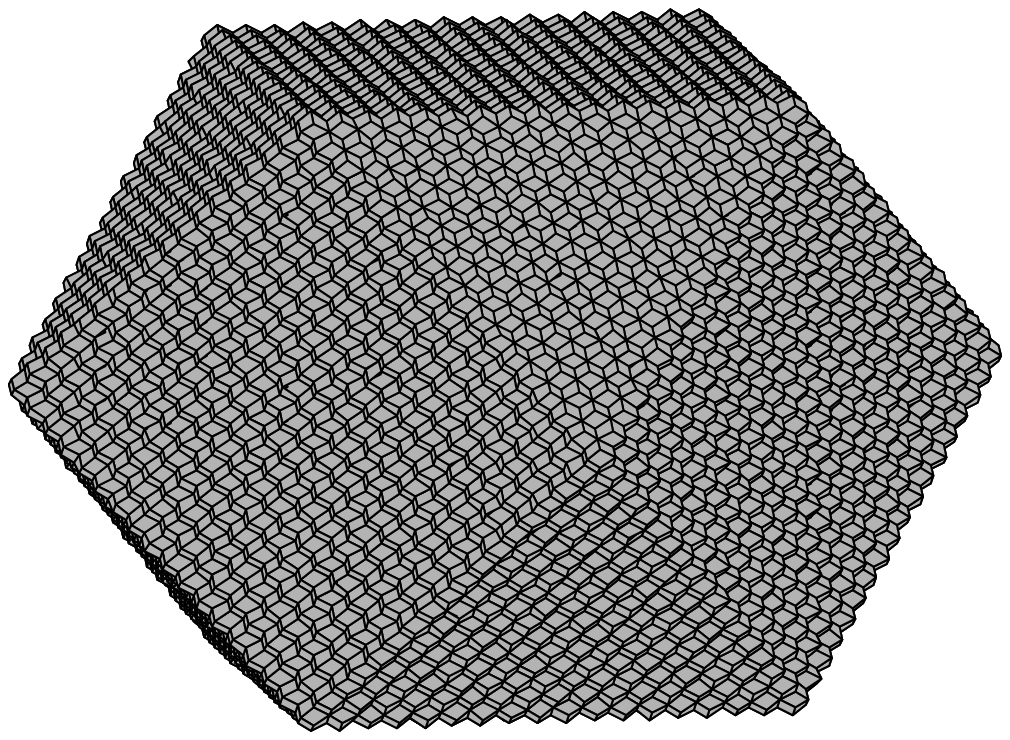}}
\hfill
\newline
\hfill
\addtocounter{subfigure}{-4}
\subfigure[]{\includegraphics[width=0.35\linewidth]{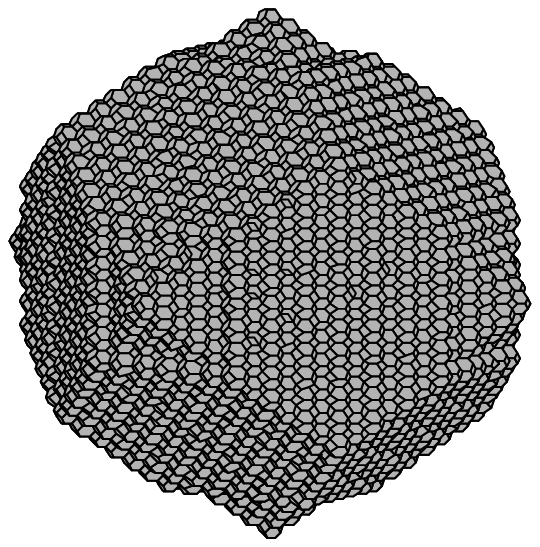}}
\hfill
\addtocounter{subfigure}{+3}
\subfigure[]{\includegraphics[width=0.35\linewidth]{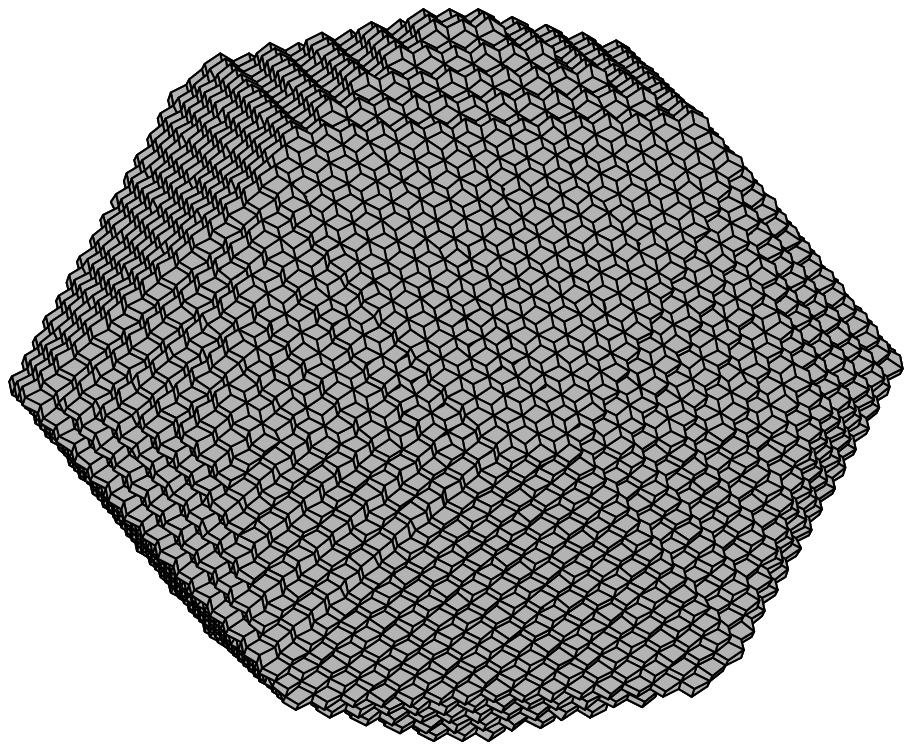}}
\hfill
\newline
\hfill
\addtocounter{subfigure}{-4}
\subfigure[]{\includegraphics[width=0.35\linewidth]{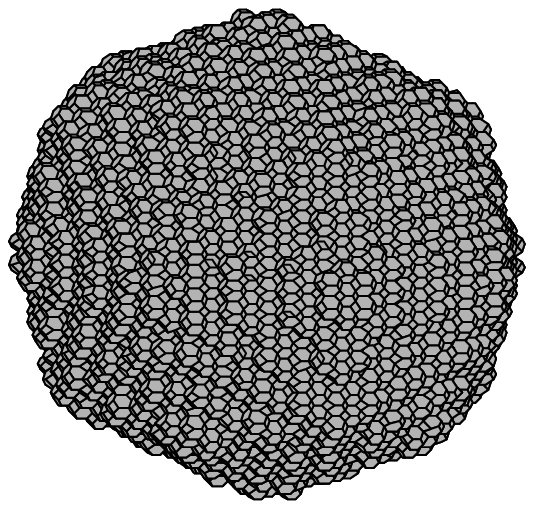}}
\hfill
\addtocounter{subfigure}{+3}
\subfigure[]{\includegraphics[width=0.35\linewidth]{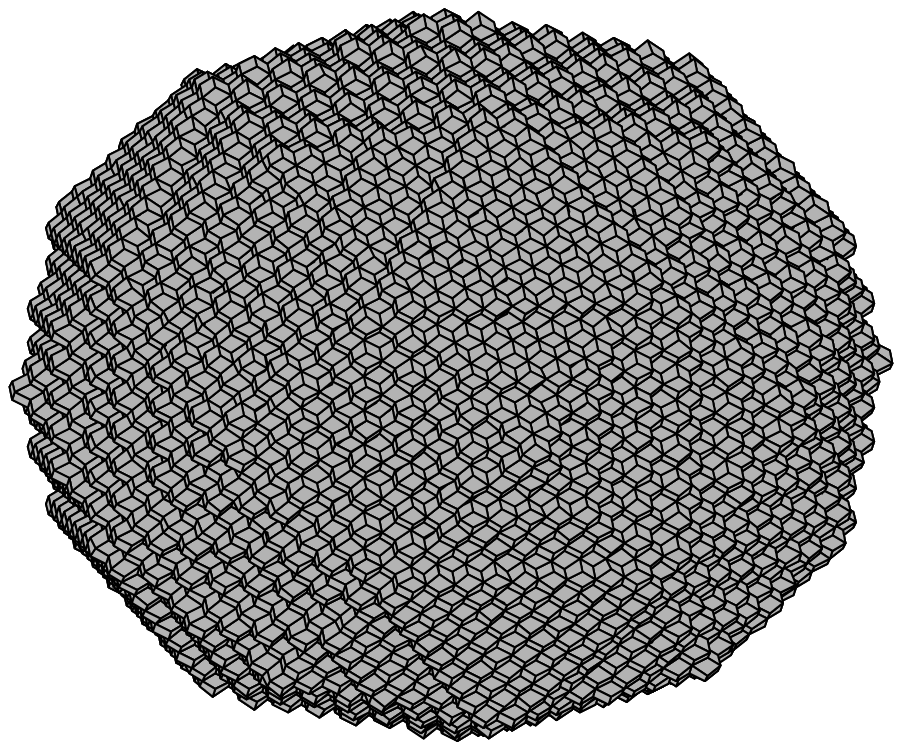}}
\hfill
\newline
\hfill
\addtocounter{subfigure}{-4}
\subfigure[]{\includegraphics[width=0.35\linewidth]{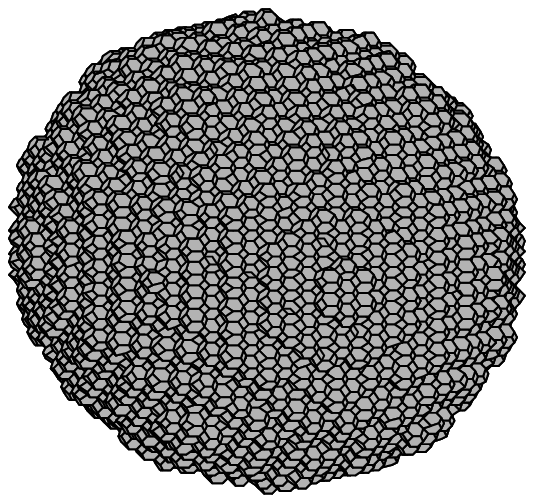}}
\hfill
\addtocounter{subfigure}{+3}
\subfigure[]{\includegraphics[width=0.35\linewidth]{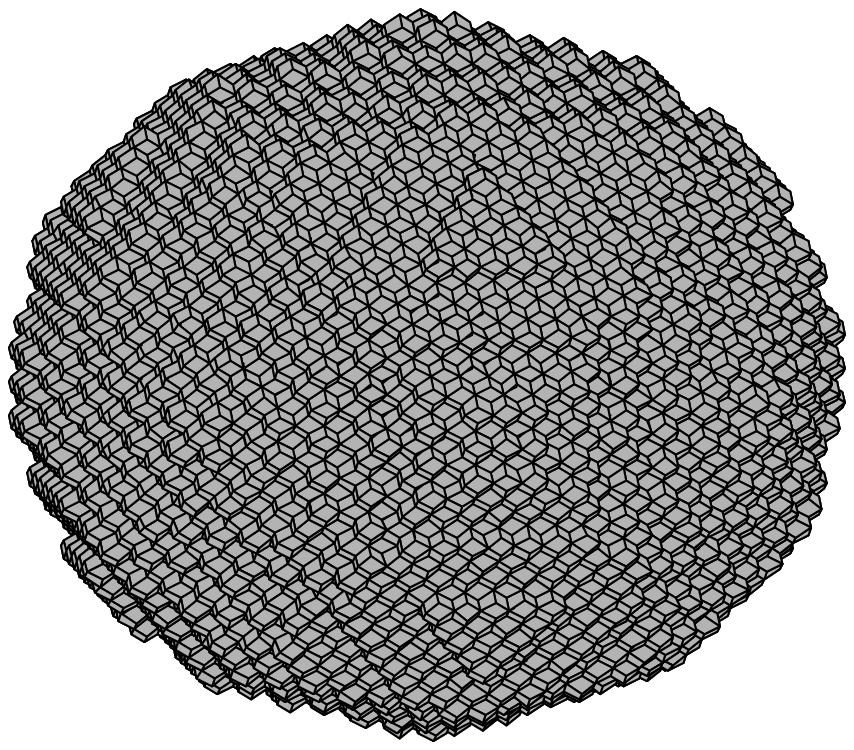}}
\hfill
\caption{The balls are constructed using the chamfer algorithm in
Section~\ref{sec:chamferAlgo}. Figure (a)--(d) show BCC balls with one
(a), two (b), three (c), and four (d) weights. Figure
(e)--(h) show FCC balls with one (e), two (f), three (g), and four (h)
weights. The weights written in bold in the tables are
used. All balls are of radius 20. \label{balls}}
\end{center}
\end{figure}

\section{Conclusions}
We have presented a general theory for weighted distances. This allows
application of the weighted distance transform to any point-lattice. 
The optimal weight calculation and the construction of a chamfer
algorithm to produce distance maps are straight-forward by the
procedure in Section~\ref{sec:coeffsopt} and
Theorem~\ref{chamf-alg-theorem}.
For the chamfer algorithm to construct correct distance-maps, the only
limitation is that the image must be wedge-preserving,
Definition~\ref{def:wedge_pres}, or have all border points in the
background, Definition~\ref{def:borderpoint}. With these conditions
which can easily be satisfied, images in any dimension can be
considered. It is also worth mentioning that despite its popularity,
to our knowledge, no proof of the correctness of the two-scan
algorithm has been published until now, except for the 2D 3 $\times$ 3
mask algorithm proof from 1966 \cite{rosenfeld:acm:1966}. 
Also, the important general formula for the
weighted distance between two grid points in
Theorem~\ref{theorem:value} is new. 

The application of these properties and of the chamfer algorithm are
numerous for the cubic and parallelepiped grids in the literature. To
illustrate how this theory can be applied, Section~\ref{sec:bccfcc} 
gives examples for the FCC and BCC grids.
But it can be applied to any point-lattice, for example the 3D
grid where the voxels are hexagonal cylinders as suggested in
\cite{brimkov:comp:2005} or 4D a grid as in \cite{borgefors:DAM:2003},
or even other grids as long as they satisfy the module conditions.


\end{document}